\documentclass[12pt]{article}
\usepackage{setspace}
\usepackage[margin = .94in]{geometry}

\usepackage[authoryear]{natbib}
\usepackage{amssymb}
\usepackage{color}
\usepackage{xcolor}
\usepackage{amsmath}
\usepackage{soul}
\usepackage{bm}
\usepackage{blindtext}
\usepackage{dsfont}
\usepackage{graphicx}
\usepackage{algorithm}
\usepackage{algpseudocode}
\usepackage{todonotes}
\usepackage[export]{adjustbox}
\usepackage{pdfpages}

\definecolor{pal1}{cmyk}{.90,.30,.0,.0}
\definecolor{pal2}{cmyk}{0.80, 0.00, 1.00, 0.00}
\definecolor{pal3}{cmyk}{0.00, 0.40, 0.25, 0.00}
\definecolor{pal4}{cmyk}{0.10,0.90,0.80,0.00}
\definecolor{palg}{HTML}{97F99F}
\definecolor{palb}{HTML}{336bc3}
\definecolor{palorange}{HTML}{FFEC8C}

\DeclareRobustCommand{\hlb}[1]{#1}

\newcommand{\Poisson}{\operatorname{Poisson}}

\newcommand{\Exponential}{\operatorname{Exponential}}

\newcommand{\Beta}{\operatorname{Be}}
\newcommand{\Uniform}{\operatorname{Uniform}}
\newcommand{\Normal}{\operatorname{Normal}}
\newcommand{\Cauchy}{\operatorname{Cauchy}}

\newcommand{\Dirichlet}{\mathcal D}

\newcommand{\Var}{\operatorname{Var}}

\newcommand{\Data}{\mathcal D}

\newcommand{\sC}{\mathcal C}

\newcommand{\Geometric}{\operatorname{Geometric}}

\newcommand{\Reals}{\mathbb R}
\newcommand{\Tree}{\mathcal T}
\newcommand{\sM}{\mathcal M}

\newcommand{\indep}{\stackrel{\text{indep}}{\sim}}

\usepackage[utf8]{inputenc} 
\usepackage[T1]{fontenc}    
\RequirePackage[colorlinks,citecolor=blue,urlcolor=blue]{hyperref}
\usepackage{url}            
\usepackage{booktabs}       
\usepackage{amsfonts}       
\usepackage{nicefrac}       
\usepackage{microtype}      
\usepackage{amsthm}

\usepackage{xcolor}

\newtheorem{theorem}{Theorem}

\newtheorem{lemma}[theorem]{Lemma}

\theoremstyle{definition}

\newtheorem{remark}{Remark}

\begin{document}

\title{\vspace{-3em} 
    Bayesian Regression Tree Ensembles that Adapt to Smoothness and Sparsity}
\author{Antonio R. Linero\thanks{Department of Statistics, Florida State
    University, Email: arlinero@stat.fsu.edu} \ and Yun
  Yang\thanks{Department of Statistics, Florida State University, Email: yyang@stat.fsu.edu}
}

\maketitle

\begin{abstract}
  Ensembles of decision trees are a useful tool for obtaining for obtaining flexible estimates of regression functions. Examples of these methods include gradient boosted decision trees, random forests, and Bayesian CART. Two potential shortcomings of tree ensembles are their lack of smoothness and vulnerability to the curse of dimensionality. We show that these issues can be overcome by instead considering sparsity inducing soft decision trees in which the decisions are treated as probabilistic. We implement this in the context of the Bayesian additive regression trees framework, and illustrate its promising performance through testing on benchmark datasets. We provide strong theoretical support for our methodology by showing that the posterior distribution concentrates at the minimax rate (up-to a logarithmic factor) for sparse functions and functions with additive structures in the high-dimensional regime where the dimensionality of the covariate space is allowed to grow near exponentially in the sample size. Our method also adapts to the unknown smoothness and sparsity levels, and can be implemented by making minimal modifications to existing BART algorithms.


  \vspace{1em}

\end{abstract}

\doublespacing

\section{Introduction}

\hlb{Consider a nonparametric regression model $Y = f_0(X) + \epsilon$ with response $Y$, $X \in [0,1]^p$ a $p$-dimensional predictor, $f_0$ an unknown regression function of interest, and Gaussian noise $\epsilon \sim \Normal(0, \sigma^2)$. Suppose we observe $\Data = ((X_1, Y_1), \ldots, (X_n, Y_n))$ consisting of independent and identically distributed copies of $(X,Y)$.} 
A popular approach to estimating \(f_0(x)\) is to form an ensemble of decision trees; common techniques include boosted decision trees \citep{freund1999short} and random forests \citep{breiman2001random}. Bayesian tree-based models, such as the Bayesian additive regression trees (BART) model \citep{chipman2010bart}, have recently attracted interest from practitioners due to their excellent empirical performance and natural uncertainty quantification; BART has been applied in a wide variety of contexts such as nonparametric function estimation with variable selection \citep{bleich2014variable, linero2016bayesian}, analysis of loglinear models \citep{murray2017log}, and survival analysis \citep{sparapani2016nonparametric}. 
\hlb{Additionally, BART is consistently among the best performing methodologies in the Atlantic Causal Inference Conference Data Analysis Challenge \citep{hill2011bayesian, Hill:2016, hahn2017bayesian, dorie2017automated}.}

Despite the recent popularity of Bayesian tree-based models, they suffer from several drawbacks. First, in the regression setting, estimators based on decision trees are not capable of adapting to higher smoothness levels exhibited in $f_0$ due to their piecewise-constant nature. Second, as illustrated by \citet{linero2016bayesian}, they suffer from the curse of dimensionality --- their prediction performance deteriorates as the dimensionality $p$ increases. Last but not least, very little theoretical work has been done for understanding large sample properties of Bayesian tree-based approaches from a frequentist perspective.

In this article, we propose a new method, called soft Bayesian additive regression trees (SBART) which improves both practically and theoretically upon existing Bayesian sum-of-trees models. To address the first aforementioned drawback, we employ a ensemble of carefully designed ``soft'' decision trees as building blocks in the BART model, and show in both empirical studies and theoretical investigation that the resulting Bayesian approach can adapt to the unknown smoothness level of the true regression function $f_0$ --- the corresponding posterior distribution achieves the minimax rate $n^{-\alpha/(2\alpha + p)}$ of contraction up to logarithmic terms \citep{ghosal2000convergence} when \(f_0 \in \mathcal C^{\alpha,R}([0,1]^d)\) where \(C^{\alpha,R}([0,1]^d)\) denotes a H\"older space with smoothness index \(\alpha\) and radius \(R\). 

To overcome the curse of dimensionality, we specify sparsity inducing priors \citep{linero2016bayesian} for the splitting rule probabilities in the soft decision trees. We show that SBART takes advantage of structural sparsity in the true regression function $f_0$ --- when $f_0$ only depends on $d \ll p$ predictors and is $\alpha$-H\"older smooth, the resulting posterior distribution contracts towards the truth at a rate of $n^{-\alpha/(2\alpha + d)}+\sqrt{n^{-1} d\log p}$ up to logarithmic terms, which is near minimax-optimal even in the high-dimensional setting where \(p\) grows nearly exponentially fast in \(n\) \citep{yang2015minimax}. Furthermore, due to the additive nature of sum-of-trees based models, we show that SBART can also adapt to low-order non-linear interactions: if $f_0$ can be decomposed into many low dimensional pieces $f_0 = \sum_{v = 1}^V f_{0v}$, where each additive component $f_{0v}$ is $d_v$-sparse and $\alpha_v$-smooth, then SBART also achieves a near-minimax rate of posterior contraction. Compared to the rate for the general sparse case, which allows at most $o(\log n)$ many active predictors for consistency, the rate for additive structures potentially allows $o(n^\beta)$ many predictors for some $\beta \in (0,1)$; this partly explains the empirical success of Bayesian sum-of-tree approaches, as many real-world phenomena can be explained in terms of a small number of low-order interactions.


Our proofs involve a key lemma that links sum-of-tree type estimators with kernel type estimators. Unlike frequentist kernel type estimators that require prior knowledge on the smoothness level of $f_0$ for choosing a smoothness matching kernel, Bayesian sum-of-tree based methods are adaptive, requiring no prior knowledge of the smoothness levels $\{\alpha_v\}$, number of additive components $V$, or degree of lower-order interactions $d_v$, while still attaining near-minimax rates even under the high-dimensional setting. Practically, SBART can be implemented by making minimal modifications to existing strategies for fitting Bayesian tree-based models: the sparsity-inducing prior uses conditionally-conjugate Dirichlet priors which can be easily accommodated during Gibbs sampling, while replacing the usual decision trees with soft decision trees requires minor changes to the backfitting algorithm typically used with BART.

\subsection{Related Work}

There has been a recent surge of interest in the theoretical properties of BART-type models. While our work was under review we learned that, \hlb{in essentially simultaneous work}, \citet{rockova2017posterior} established similar posterior contraction rates for a particularly designed BART prior, using a so-called ``spike-and-tree'' prior to allow for the ensemble to adapt to sparsity. In particular, they show that a single deep decision tree can approximate any function with smoothness level $\alpha\leq 1$, which is then divided among trees with smaller depth. 
Our theory instead relies on linking sum-of-tree type estimators with kernel type estimators, which only need shallow trees and motivate the usage of soft-decision trees. Practically, the most relevant difference is that our SBART prior allows for adaptation to the smoothness level even when \(\alpha > 1\), whereas the use of piecewise-constant basis functions in traditional BART models only allows for adaptation to functions which are at-most Lipschitz-smooth ($\alpha \le 1$). 
An additional difference is that we focus on establishing concentration results for the fractional posterior, which allows for less restrictive assumptions about our choice of prior; in our supplementary material, we also provide concentration results for the usual posterior, under more stringent conditions. In even more recent work, \citet{alaa2017bayesian} establish consistency results for BART-type priors for estimating individual treatment effects in causal inference settings, and also noted the limitation of BART in adapting to a smoothness order higher than \(\alpha =1\).

The soft decision trees we use are similar in spirit to those used by \citet{irsoy2012soft}, who considered a soft variant of the CART algorithm. Our work differs in that (i) our trees are not learned in a greedy fashion, but instead by extending the Bayesian backfitting approach of \citet{chipman2010bart}, (ii) we consider an ensemble of soft trees rather than a single tree, (iii) we use a different parameterization of the gating function which does not consider oblique decision boundaries, and (iv) we establish theoretical guarantees for our approach.



The rest of the paper is organized as follows. In Section~\ref{sec:SBART}, we develop our SBART prior. In Section~\ref{sec:theory} we state our theoretical results. In Section~\ref{sec:illustrations}, we illustrate the methodology on both simulated and real datasets. We finish in Section~\ref{sec:discussion} with a discussion. Proofs are deferred to the appendix. In supplementary material, we provide additional computational details, timing results, and additional theoretical results extending our fractional posterior results to the usual posterior.

\section{Soft Bayesian sum of trees models}
\label{sec:SBART}

\subsection{Description of the model}
\label{sec:description}

We begin by describing the usual ``hard'' decision tree prior used in BART. We model \(f_0(x)\) as the realization of a random function
\begin{align}\label{Eqn:SumOfTree}
  f(x) = \sum_{t = 1}^T g(x; \Tree_t, \sM_t), \qquad x \in \Reals^p,
\end{align}
where \(\Tree_t\) denotes the topology/splitting rules of the tree, \(\sM_t = (\mu_{t1}, \ldots, \mu_{tL_t})\) is a collection of parameters for the leaf nodes and \(L_t\) denotes the number of leaves. The function \(g(x ; \Tree_t, \sM_t)\) returns \(\sum_{\ell = 1}^L \mu_{t\ell}\, \phi(x ; \Tree_t, \ell)\) where \(\phi(x ; \Tree_t, \ell)\) is the indicator that \(x\) is associated to leaf node \(\ell\) in $\Tree_t$.

Following \citet{chipman2010bart}, we endow \(\Tree_t\) with a branching process prior. The branching process begins with a root node of depth \(k = 0\). For \(k = 0, 1, 2, \ldots\), each node at depth \(k\) is non-terminal with probability \(q(k) = \gamma(1 + k)^{-\beta}\) where \(\gamma > 0\) and \(\beta > 0\) are hyperparameters controlling the shape of the trees. It is easy to check using elementary branching process theory that this process terminates almost surely provided that $\beta > 0$ \citep{athreya2004branching}.

Given the tree topology, each branch node \(b\) is given a decision rule of the form \([x_j \le C_b]\), with \(x\) going left down the tree if the condition is satisfied and right down the tree otherwise. The predictor \(j\) is selected with probability \(s_j\) where \(s = (s_1, \ldots, s_p)\) is a probability vector. We assume that \(C_b \sim \Uniform(a,b)\) where \(a\) and \(b\) are chosen so that the cell of \(\Reals^p\) defined by the path to \(b\) is split along the \(j\)th coordinate. The leaf parameters \(\mu_{t\ell}\) are assumed independent and identically distributed from a \(\Normal(0, \sigma_\mu^2/T)\) distribution. The scaling factor \(T\) ensures the stability of the prior on \(f\) as the number of trees increases --- loosely speaking, the functional central limit theorem implies the convergence of the prior on \(f\) to a Gaussian process as \(T \to \infty\).

We now describe how to convert the hard decision tree described above into a soft decision tree. Rather than \(x\) following a deterministic path down the tree, \(x\) instead follows a probabilistic path, with \(x\) going left at branch \(b\) with probability
\begin{math}
  \psi(x ; \Tree, b) =
  \psi\left( \frac{x_j - C_b}{\tau_b} \right),
\end{math}
where $\tau_b>0$ is a bandwidth parameter associated with branch $b$. Averaging over all possible paths, the probability of going to leaf $\ell$ is 
\begin{align}
  \label{eq:proper}
  \phi(x ; \Tree, \ell) = \prod_{b \in A(\ell)} \psi(x; \Tree, b)^{1 - R_b} (1 - \psi(x; \Tree, b))^{R_b},
\end{align}
where \(A(\ell)\) is the set of ancestor nodes of leaf \(\ell\) and \(R_b = 1\) if the path to \(\ell\) goes right at \(b\). The parameter \(\tau_b\) controls the the sharpness of the decision, with the model approaching a hard decision tree as \(\tau_b \to 0\), and approaching a constant model as \(\tau_b \to \infty\). Unlike hard decision trees where each leaf is constrained to only locally influence the regression function $f$ near its center $\{C_b\}$, each leaf in the soft decision tree imposes a global impact on $f$, whose influence as $x$ deviates from the center depends on the local bandwidths $\{\tau_b\}$. As we will illustrate, this global impact of local leaves enables the soft tree model to adaptively borrow information across different covariate regions, where the degree of smoothing is determined by the local bandwidth parameters learned from the data. This is illustrated in Figure~\ref{fig:gate-to-kernel} for a simple univariate soft decision tree. In our illustrations we use the logistic gating function $\psi(x) = (1 + e^{-x})^{-1}$.



\begin{figure}[t]
  \centering
  \includegraphics[width = .5\textwidth,valign=t]{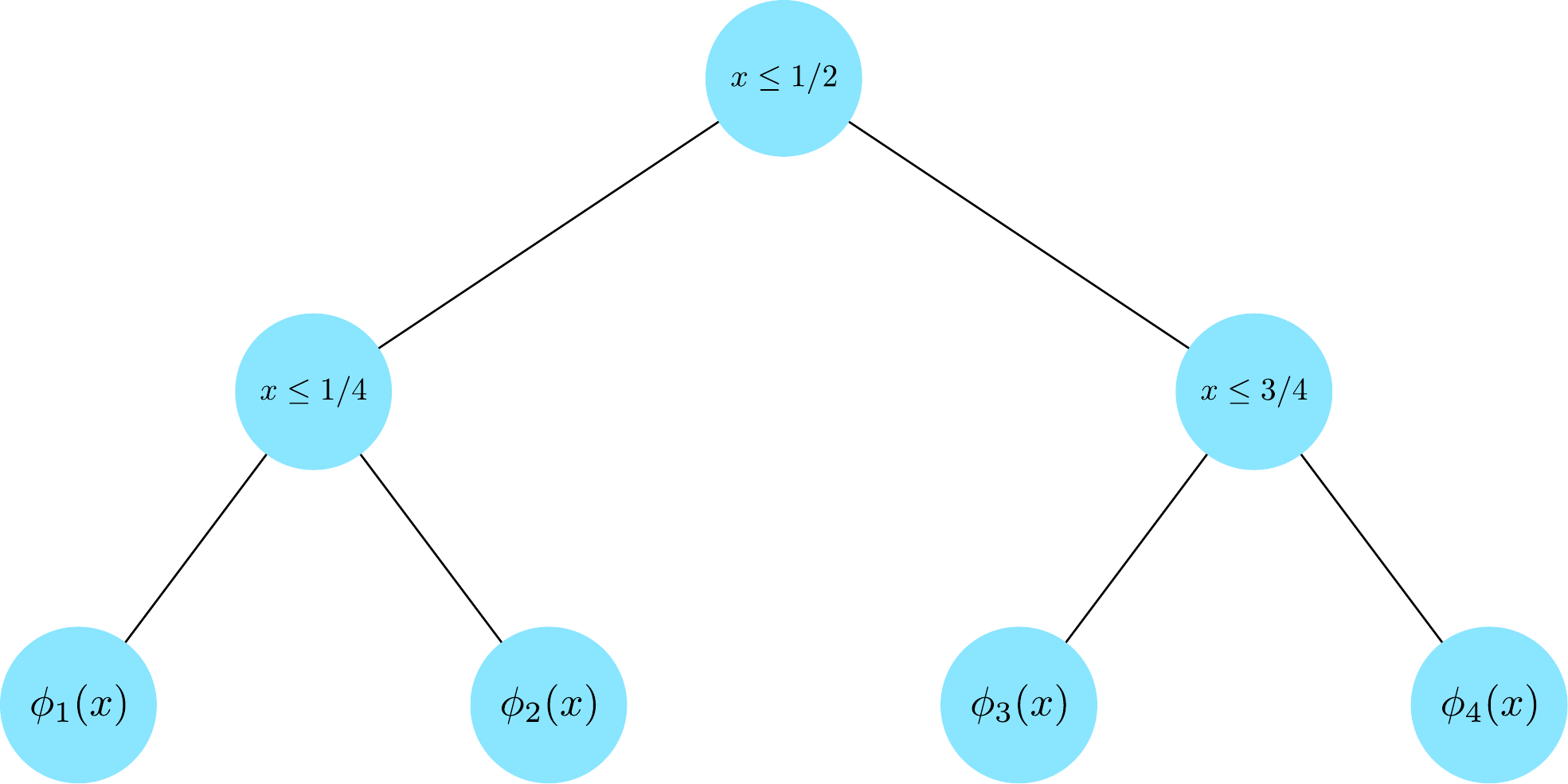}
  \includegraphics[width = .4\textwidth,valign=t]{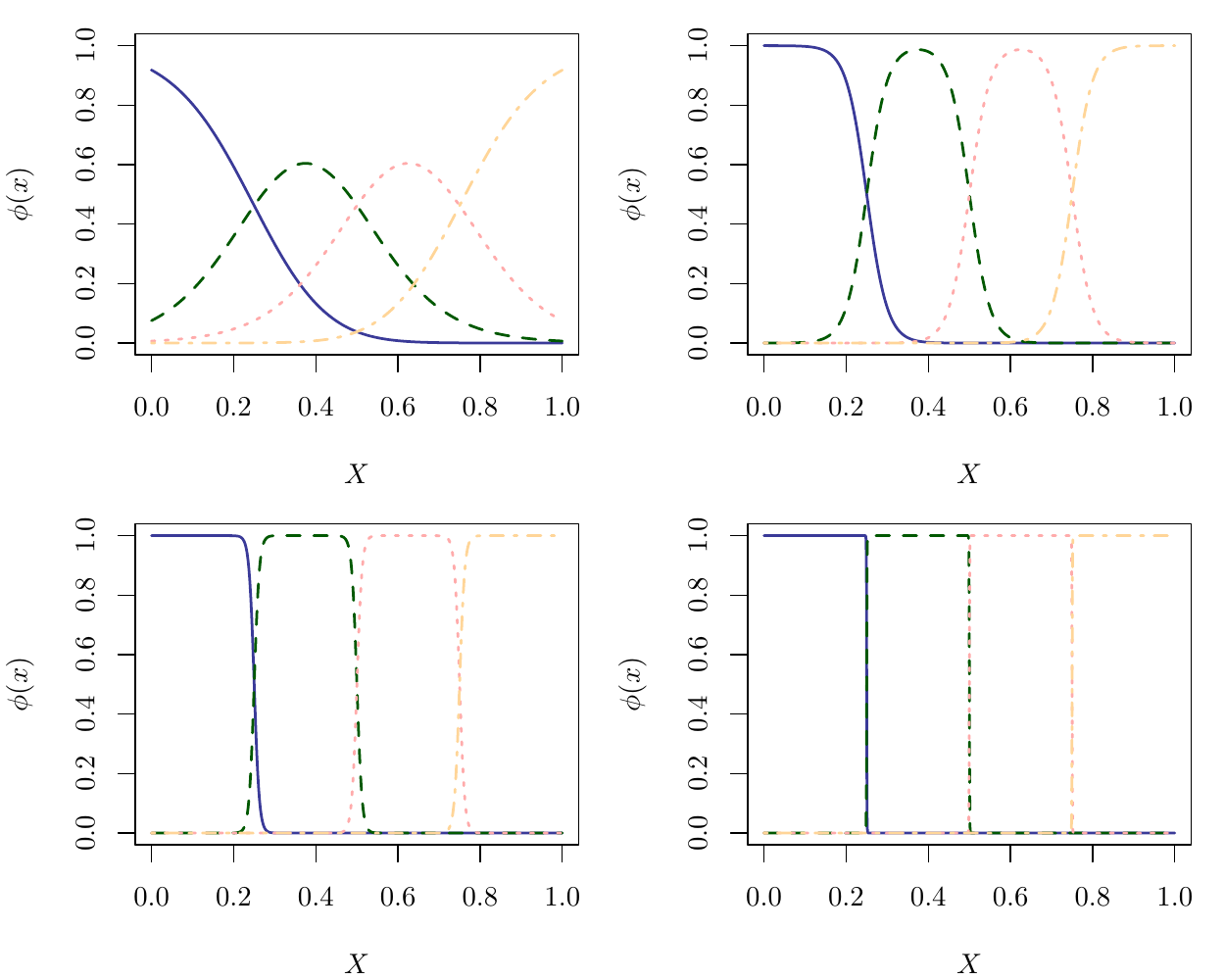}
  \caption{Left: example tree, with cut points at $x = 0.5, 0.25$, and $0.75$.
    Right: the weights $\phi_\ell(x)$ for $\ell = 1,\ldots,4$ as functions of
    $x$ for the values $\tau^{-1} \in \{10,40,160,2560\}$.}
  \label{fig:gate-to-kernel}
\end{figure}

\subsection{Smoothness adaptation}

A well-known feature of decision trees is their lack of smoothness. Single-tree algorithms such as the CART algorithm \citep[Chapter 9.2]{hastie2016elements} result in step-function estimates, suggesting that they should not be capable of efficiently estimating smooth functions \citep{gyorfi2006distribution}. Methods based on ensembles of decision trees average over many distinct partitions of the predictor space, resulting in some degree of smoothing. Even with this averaging, the estimated regression functions are not smooth. Heuristically, we note that under our BART specification the function $f$ is not differentiable in quadratic mean. Indeed, with trees of depth $1$, $p = 1$, and cutpoints $C_b \sim G$, simple calculations give $E\{(f(x + \delta) - f(x))^2\} \propto \delta G'(x) + o(\delta)$. Consequently, BART ensembles with a large number of trees resemble nowhere-differentiable continuous functions, and in the limit as $T \to \infty$ the BART prior converges to a nowhere-differentiable Gaussian process. This heuristic argument suggests that BART can only adapt to functions with H\"{o}lder smoothness level no greater than $\alpha = 1$ (Lipschitz functions).

Figure~\ref{fig:simple-functions} compares the fit of BART to SBART with \(\tau_b \equiv 0.1\). We see that when $T = 1$ trees are used we require a large number of leaf nodes to model relatively simple functions. At a large scale, we see that the BART fit resembles a nowhere-differentiable continuous function. While an improvement, the estimate from BART is still not sufficiently smooth and exhibits large fluctuations.

The fit of the soft decision tree in Figure~\ref{fig:simple-functions} by comparison is infinitely differentiable and requires only a small number of parameters. Consequently, we obtain a fit with lower variance and negligible bias. An attractive feature of soft decision trees exhibited in Figure~\ref{fig:simple-functions} is their ability to approximate linear relationships. In this case, even when $T = 1$, we recover the smooth functions almost exactly.

\begin{figure}[t]
  \centering
  \includegraphics[width=.8\textwidth]{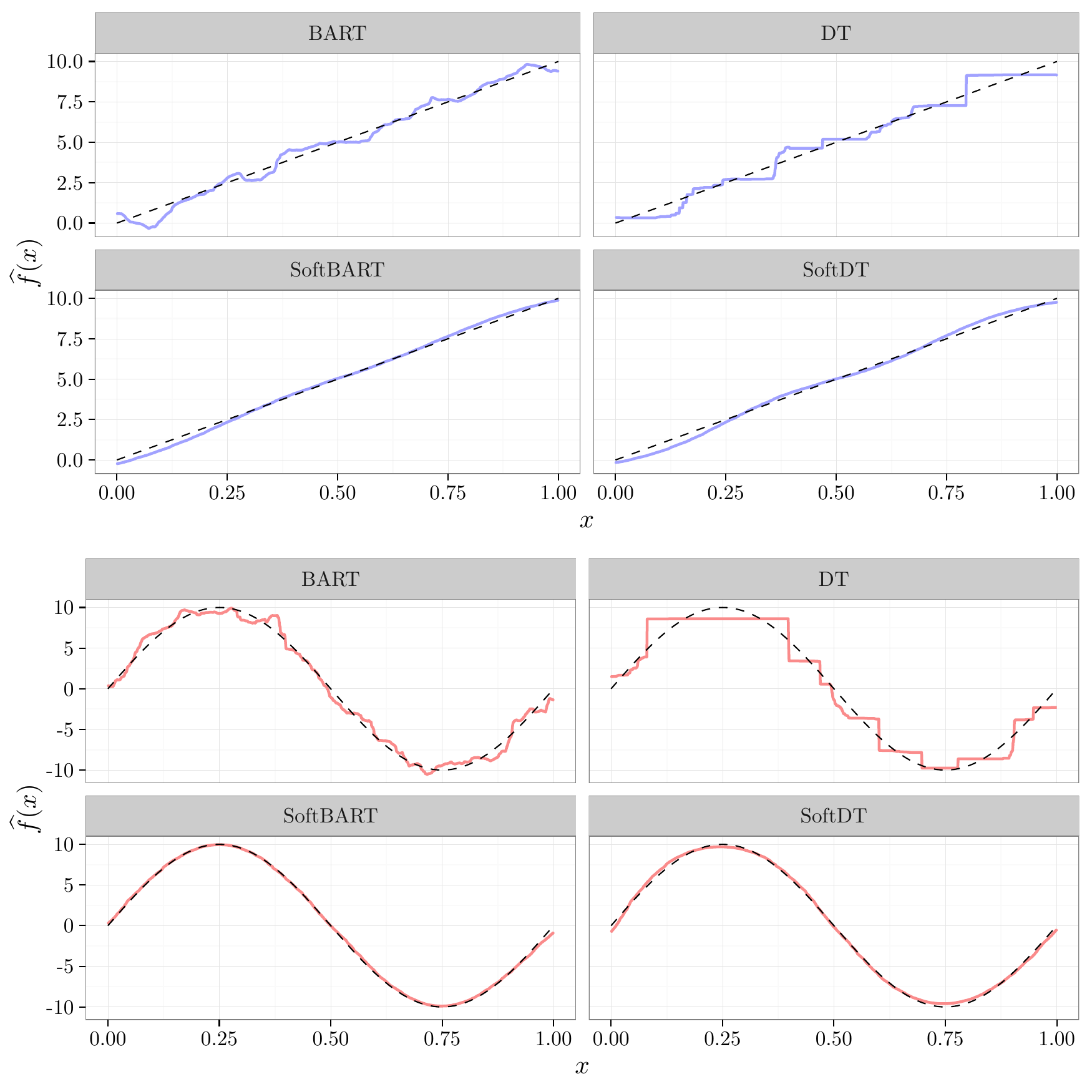}
  \caption{Posterior means (solid) against underlying true regression function
    (dashed). Error variance is \(\sigma^2 = 2^2\). Top: \(f(x) = 10 x_1\).
    Bottom: \(f(x) = 10 \sin(2 \pi x_1)\). BART denotes the BART model with \(T
    = 50\), DT denotes the BART model with \(T = 1\), and soft variants are
    prefixed by Soft.}
  \label{fig:simple-functions}
\end{figure}

\subsection{Prior specification and implementation}
\label{sec:default-prior}

Following \citet{chipman2010bart}, in this section we develop a ``default'' SBART prior. The goal is to develop a prior which can be used routinely, without requiring the user to specify any hyperparameters; while the choices below may appear ad-hoc, they have been found to work remarkably well across a wide range of datasets. After adopting the following default prior, users may wish to further tune the number of trees \(T\), the parameter \(r\) in the prior for \(\tau_b\), or use additional information regarding the targeted sparsity level. We stress, however, that a reasonable baseline level of performance is obtained without the need to do any further tuning.

Following \citet{chipman2010bart}, we recommend scaling \(Y\) so that most/all of the responses fall in the interval \([-0.5, 0.5]\). We also preprocess \(X_j\) so that \(X_j \sim \Uniform(0,1)\) approximately by applying a quantile normalization in which each \(X_{ij}\) is mapped to its rank, with $\min X_{ij} = 1$ and $\max X_{ij} = n$. We then apply a linear transformation so that the values of \(X_{ij}\) are in \([0,1]\). The goal of this preprocessing of \(X\) is to make the prior invariant under monotone transformations of \(X\), which is a highly desirable property of the original default BART model. 


We now describe our default prior for the bandwidths \(\tau_b\) and the splitting proportions \(s = (s_1, \ldots, s_p)\). We use a sparsity-inducing Dirichlet prior,
\begin{align}
  \label{eq:sparsity-inducing}
  s \sim \Dirichlet(a / p^\xi, \ldots, a / p^\xi), \qquad \xi \ge 1.
\end{align}
Our theoretical results require \(\xi > 1\), however in practice we find that setting \(\xi = 1\) works adequately. This Dirichlet prior for \(s\) was introduced by \citet{linero2016bayesian}; throughout, we refer to the BART model with \eqref{eq:sparsity-inducing} as Dirichlet additive regression trees (DART) to contrast with BART when no such sparsity-inducing prior is used. The parameter \(a\) controls the expected amount of sparsity in \(f\). Conditional on there being \(B\) branches in the ensemble, the number of predictors included in the ensemble is converges in distribution to $1 + Z$ where $Z \sim \Poisson(\theta)$ and $\theta = a\sum_{i = 1}^{B - 1} (a + i)^{-1}$ \citep{linero2016bayesian} when \(\xi = 1\). When prior information is available on the sparsity of \(f_0\), we can choose \(a\) to match the targeted amount of sparsity. By default we use a compound Gamma prior,
\begin{math}
  a / (a + \lambda_a) \sim \Beta(a_a, b_a),
\end{math}
with \(a_a = 0.5, b_a = 1, \lambda_a = p\). This prior attempts to strike a balance between the sparse and non-sparse settings by having an infinite density at \(0\), median \(\alpha = p / 4\), and an infinite mean.

There are several possibilities for choosing the bandwidth \(\tau_b\). In preliminary work, using tree-specific $\tau_t$'s shared across branches in a fixed tree worked well, with $\tau_t \sim \Exponential(r)$ where \(E(\tau_t) = r\). Our illustrations use $r = 0.1$, which, as shown in Figure~\ref{fig:gating}, gives a wide range of possible gating functions. An interesting feature of the sampled gating functions is that both approximate step functions and approximately linear functions are supported.
\begin{figure}
  \centering
  \includegraphics[width = .55\textwidth]{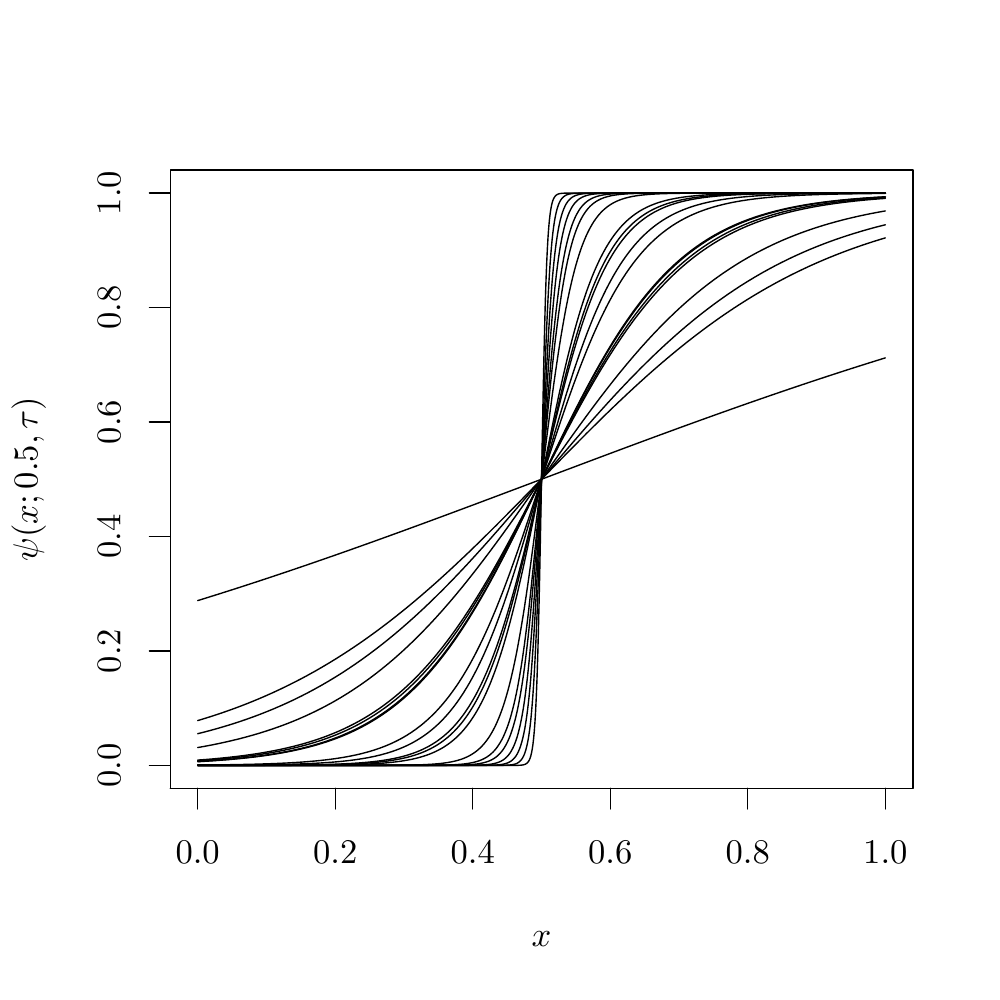}
  \caption{Draws of the gating function \(\psi(x ; \Tree, b)\) when \(\tau_b
    \sim \Exponential(0.1)\) and \(C_b = 0.5\).}
  \label{fig:gating}
\end{figure}

We give \(\sigma = {\Var(\epsilon)}^{\nicefrac 1 2}\) a half-Cauchy prior, \(\sigma \sim \Cauchy_+(0, \widehat \sigma)\). Again following \citet{chipman2010bart}, $\widehat \sigma$ is an estimate of \(\sigma\) based on the data. We use an estimate \(\widehat \sigma_{\text{lasso}}\) of \(\sigma\) obtained by fitting the lasso using the \texttt{glmnet} package in \texttt{R}.

The model has hyperparameters \((\sigma_\mu^2, \gamma, \beta, T)\). In preliminary work, we did not have success placing priors on \(\gamma\) and \(\beta\), and instead fix \(\gamma = 0.95\) and \(\beta = 2\) \citep{chipman2010bart}. We give \(\sigma_{\mu}\) a half-Cauchy prior, \(\sigma_{\mu} \sim \Cauchy_+(0, 0.25)\), where \(0.25\) is chosen so that \(\sigma_{\mu}\) has median equal to the default value recommended by \citet{chipman2010bart}.

An important remaining specification is the number of trees \(T\) to include in the ensemble. The theoretical results we establish in Section~\ref{sec:theory} make use of a prior distribution on \(T\); however, our attempts to incorporate a prior on \(T\) using reversible jump methods \citep{green1995reversible} resulted in poor mixing of the associated MCMC algorithms. Generally, we have found that fixing \(T\) at a default value of \(T = 50\) or \(T = 200\) is sufficient to attain good performance on most datasets. Tuning \(T\) further often provides a modest increase in performance, but may be worth the effort on some datasets (see Section~\ref{sec:benchmark}).

There are a number of possible options for tuning \(T\), such as approximate leave-one-out cross validation using Pareto-smoothed importance sampling (PSIS-LOO) \citep{vehtari2015practical}, maximizing an approximate marginal likelihood obtained using (say) WBIC \citep{watanabe2013widely}, or \(K\)-fold cross validation as recommended by \citet{chipman2010bart}. The advantage of WBIC and PSIS-LOO is that they require fitting the model only once for each value of \(T\). In practice, we have found that approximations such as WBIC and PSIS-LOO are unreliable, with PSIS-LOO prone to overfitting and WBIC requiring potentially very long chains to estimate. Figure~\ref{fig:kfold} displays the values of PSIS-LOO, a WBIC approximation of the negative marginal likelihood of \(T\) \citep{watanabe2013widely}, and $5$-fold cross validation, when used to select $T$ for a replicate of the illustration in Section~\ref{sec:friedman} with $p = 100$ predictors. Both WBIC and cross validation select $T = 10$, which also minimizes the root mean squared error $\int (f_0(x) - \widehat f(x))^2 \ dx$. Resource permitting, we have found \(K\)-fold cross-validation to be the most reliable method for selecting \(T\). 

\begin{figure}
  \centering
  \includegraphics[width=.8\textwidth]{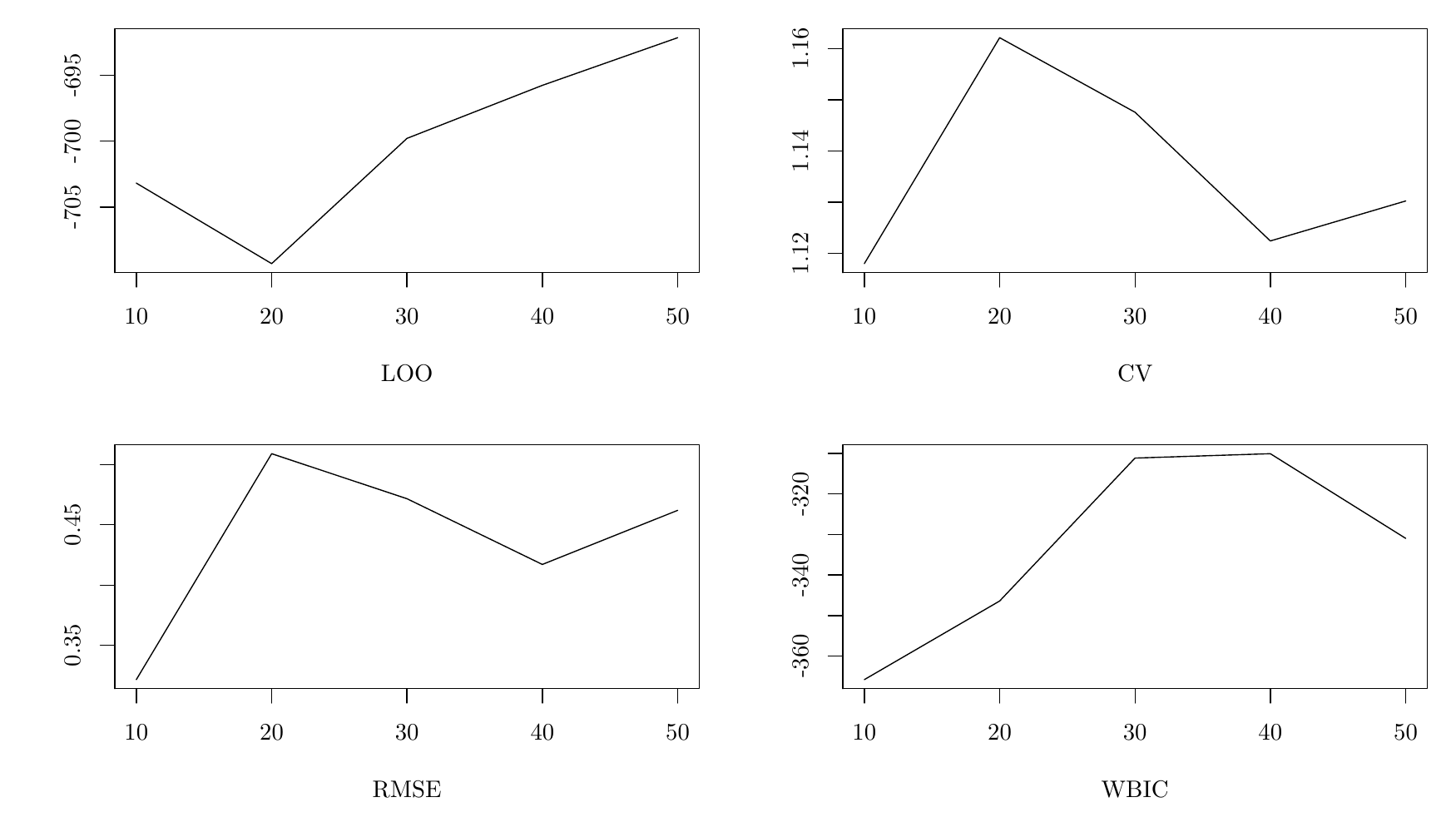}
  \caption{Selecting $T$ using LOO, cross validation, and WBIC, with the
    population root mean squared error for $f$.}
  \label{fig:kfold}
\end{figure}

As a default we use the following priors throughout the manuscript.
\begin{equation}
\begin{alignedat}{3}
  s      & \sim \Dirichlet(a / p, \ldots, a / p),   & \qquad &  &  & \frac{a}{a  + p}  \sim \Beta(0.5, 1),       \\
  \tau_t & \indep \Exponential(0.1),                          & \qquad &  &  & \sigma_{\mu }               \sim \Cauchy_+(0, 0.25) , \\
  \sigma & \sim \Cauchy_+(0, \widehat \sigma_{\text{lasso}}), & \qquad &  &  & \gamma = 0.95 ,                                \\
  \beta & = 2.                                            & 
\end{alignedat}\label{eq:def-prior}
\end{equation}

\subsection{Variable grouping prior}

The sparsity-inducing prior \eqref{eq:sparsity-inducing} can be extended to allow penalization of groups of predictors simultaneously, in a manner similar to the group lasso \citep{yuan2006model}. Suppose that the predictors can be divided into \(M\) groups of size \(P_m\). We set
\begin{align}
  \label{eq:grouping-prior}
  \begin{split}
    s_{mk} &= u_{m} \cdot v_{mk}, \\
    u &\sim \Dirichlet(a / M, \ldots, a / M), \\
    v_{m} &\sim \Dirichlet(\omega / P_m, \ldots, \omega / P_m).
  \end{split}
\end{align}
We primarily use the grouping prior to allow for the inclusion of categorical predictors through the inclusion of dummy variables. This is an extension of the approach used by the \texttt{bartMachine} package in \texttt{R}. An alternative approach to the inclusion of categorical predictors, used in the \texttt{BayesTree} package, is to construct decision rules based on a dummy variable \(Z_j = I(X_j \in A_b)\) where \(A_b\) is a random subset of the possible values of predictor \(j\). In our illustrations, we let \(\omega \to \infty\) so that \(v_{mk} = P_m^{-1}\) and set \(a / (a + M) \sim \Beta(0.5, 1)\).

\subsection{Posterior computation}

We use the Bayesian backfitting approach described by \citet{chipman2010bart} to construct a Markov chain Monte Carlo (MCMC) algorithm to sample approximately from the posterior.

\begin{algorithm}
  \caption{Bayesian backfitting algorithm} \label{alg:bayesian-backfitting}
  \begin{algorithmic}[1]
    \For{$t = 1, \ldots, T$}
    \State Set \(Y_i^\star \gets Y_i - \sum_{k \ne t} g(X ; \Tree_{k} , \sM_k)\)
    for \(i = 1, \ldots, N\).
    \State Sample \(\Tree_t \sim \operatorname{Metrop}_{\Tree}(Y^\star, X, \tau_t, h)\). 
    \State Sample \(\tau_t \sim \operatorname{Metrop}_{\tau}(Y^\star, X, \Tree_t, h)\).
    \State Sample \(\sM_t \sim \Normal(\widehat \mu_t, \Omega_t)\) with
    \((\widehat \mu_t, \Omega_t)\) described as in the supplementary material.
    \EndFor
    \State Sample \(s \sim \Dirichlet(a / p^\xi + c_1, \ldots, a / p^\xi + c_p)\)
    where \(c_j = \#\{b : \text{branch $b$ splits on predictor $j$}\}\).
    \State Sample $(\sigma, \sigma_{\mu}, a)$ as described in the
    supplementary materials.
  \end{algorithmic}
\end{algorithm}

Within Algorithm~\ref{alg:bayesian-backfitting}, \(\Tree_t\) is updated using a Metropolis-Hastings proposal. Proposals consist of one of three possible moves: \texttt{Birth}, which turns a leaf node into a branch node; \texttt{Death}, which turns a branch node into a leaf node; and \texttt{Change}, which changes the decision rule of a branch \(b\). A detailed description of these moves, and their associated transition probabilities, is given in the supplementary materials.

Constructing efficient updates for \(\Tree_t\) and \(\tau_t\) requires marginalizing over \(\sM_t\). Because the errors are assumed Gaussian, this marginalization can be carried out in closed form. The main computational drawback of SBART relative to BART lies in this marginalization, as SBART requires computing a likelihood contribution for each leaf-observation pair, whereas BART only requires a single likelihood contribution for each tree. Hence, if the trees are deep, BART will be substantially faster. By the construction of the prior, trees generally are not deep enough for this difference to be prohibitive.

The Dirichlet prior \(s \sim \Dirichlet(a / p^\xi, \ldots, a / p^\xi)\) allows for a straight-forward Gibbs sampling update, with the full conditional given by
\begin{math}
  s \sim \Dirichlet(a / p^\xi + c_1, \ldots, a / p^\xi + c_p),
\end{math}
where \(c_j = \#\{b : \text{branch $b$ splits on predictor $j$}\}\). When the grouping prior \eqref{eq:grouping-prior} is used we also obtain simple Gibbs sampling updates, with $u \sim \Dirichlet(a/M + z_1, \ldots, a/M + z_M)$ and $v_m \sim \Dirichlet(\omega / P_m + c_{m1}, \ldots, \omega / P_m + c_{mP_m})$,
where \(z_m = \#\{b : \text{branch $b$ splits on a predictor in group $m$}\}\)
and \(c_{mk} = \#\{b : \text{branch $b$ splits on predictor $mk$}\}\).


\section{Theoretical results}
\label{sec:theory}

We study the theoretical properties of the SBART procedure from a frequentist perspective by assuming that $(Y_1,Y_2,\ldots,Y_n)$ are generated from the model $Y_i=f_0(X_i)+\epsilon_i$ with some true unknown regression function $f_0$. We assume that $f_0$ is a function over $[0,1]^p$. We prove posterior consistency results when $f_0$ is a member of certain H\"older spaces. Let $\sC^{\alpha}([0,1]^p)$ denote the H\"older space with smoothness index $\alpha$, i.e., the space of functions on $[0,1]^p$ with bounded partial derivatives up-to order $\beta$, where $\beta$ is the largest integer strictly less than $\alpha$ and such that the partial derivatives of order $\beta$ are H\"older-continuous of order $\alpha - \beta$. Let $\sC^{\alpha,R}([0,1]^p) = \{f \in \sC^\alpha([0,1]^p) : \|f\|_\alpha \le R\}$ denote the H\"older-ball of radius $R$ with respect to the H\"older norm $\|f\|_\alpha$ (see \citealp{ghosal2017fundamentals}, Appendix C).

We consider the posterior convergence of the Bayesian fractional posterior obtained by raising the likelihood function by a factor \(\eta \in (0,1]\) in the Bayes formula
\begin{align}
  \label{eq:fractional}
  \Pi_{n,\eta}(A)
  =
  \frac{\int_A \prod_{i = 1}^n p_f(Y_i \mid X_i)^\eta \ \Pi(df)}
  {\int \prod_{i=1}^n p_f(Y_i \mid X_i)^\eta\ \Pi(df)},
\end{align}
where $\Pi$ denotes the prior probability measure over $\mathcal L^2([0,1]^p)$, the $\mathcal{L}^2$ space over $[0,1]^p$.
\hlb{
Fractional posteriors have gained renewed attention in Bayesian statistics due to their robustness to model misspecification \citep{grunwald2012safe,miller2018robust}. According to \cite{walker2001bayesian}, the fractional posterior can be viewed as combining the original likelihood function with a data-dependent prior that is divided by a portion of the likelihood. This data dependent reweighting in the prior helps to prevent from possible inconsistencies by reducing the weights of those parameter values that ``track the data too closely''. Additionally, the fractional posterior with $\eta < 1$ permits much simpler theoretical analyses. 
}
Note that $\eta = 1$ corresponds to the usual posterior distribution. Abusing notation slightly, we will also use $\Pi$ to denote the prior probability measure over the parameters $(\Tree_t, \sM_t)$ and any hyperparameters in the model. Our goal is to find a sequence $\{\varepsilon_n:\, n\geq 1\}$ such that, for a sufficiently large constant $M$ and fixed $\eta$, 
\begin{align}
  \label{eq:consistent}
  \Pi_{n,\eta}\big[ \|f- f_0\|_n \geq M \varepsilon_n\big] \to 0,
  \quad
  \text{in probability as $n,\,p \to \infty$},
\end{align}
where $\|\cdot\|_n$ denotes the $\mathcal{L}^2(\mathbb{P}_n)$ norm on the function space $\mathcal L^2([0,1]^p)$ defined by $\|f-g\|_n^2=n^{-1}\sum_{i=1}^n(f(X_i)-g(X_i))^2$. The sequence $\varepsilon_n$ is then an upper bound on the posterior contraction rate. The norm $\|\cdot\|_n$ is a commonly adopted discrepancy metric in function estimation problems.

In this section, we focus on establishing \eqref{eq:consistent} for $\eta < 1$. The benefit of considering $\eta < 1$ is that this allows us to bypass verifying technical conditions regarding the effective support of the prior and the existence of a certain sieve \citep{ghosal2000convergence,ghosal2007convergence}, which allows for \eqref{eq:consistent} to be established under very weak conditions. In the supplementary material we establish posterior consistency for $\eta = 1$ under more stringent conditions on the prior. 

The main condition governing the posterior contraction rate is that the prior $\Pi$ is sufficiently ``thick'' at $f_0$, in the sense that there exists a $C > 0$ such that
\begin{align}
  \label{eq:lower-bound}
  \Pi(B_{\varepsilon_n}(f_0)) \ge e^{-Cn\varepsilon_n^2},
\end{align}
where $B_\varepsilon(f_0)$ denotes an $\varepsilon$-Kullback-Leibler (KL) neighborhood of the truth
\begin{align*}
  B_\epsilon(f_0) =
  \left\{ f :  n^{-1}\sum_{i=1}^n\int p_{f_0}^{(i)} \log\left( \frac{p^{(i)}_{f_0}}{p^{(i)}_f} \right) \ dy \le \varepsilon^2   \right\}
  \cap
  \left\{ f : n^{-1}\sum_{i=1}^n\int p^{(i)}_{f_0} \log^2\left( \frac{p^{(i)}_{f_0}}{p^{(i)}_f} \right) \ dy \le \varepsilon^2 \right\},
\end{align*}
where $p_f^{(i)}$ denotes the $i$th Gaussian density with mean $f(X_i)$ and variance $\sigma^2$. For convenience, we adopt the customary practice of assuming that $\sigma$ is fixed and known when studying the posterior contraction rate. In the regression setting, it is straightforward to verify that the KL neighborhood $B_\varepsilon(f_0)$ contains the $\mathcal{L}^2(\mathbb{P}_n)$ neighborhood $\{\|f-f_0\|_n \leq 2\sigma\, \varepsilon\}$. Therefore, to establish condition~\eqref{eq:lower-bound}, it suffices to find $\varepsilon_n$ such that $\Pi(\|f-f_0\|_\infty \leq 2\sigma\, \varepsilon_n) \geq e^{-Cn\varepsilon_n^2}$ holds, where $\|g\|_\infty=\sup_{x\in[0,1]^p}|g(x)|$ denotes the $\sup$ norm of a function $g$ in $\mathcal{L}^2([0,1]^p)$.

We establish \eqref{eq:lower-bound} for a wide class of tree-based models by deriving sharp small-ball probabilities in the $\|\cdot\|_\infty$ norm around the true regression function $f_0$. To be general, we consider any gating function $\psi:\,\mathbb{R}\to \mathbb{R}$ satisfying the following assumption.


\paragraph{{\bf Assumption G} (gating function):} Let $K=\psi(1-\psi)$ be an ``effective'' kernel function associated with gating function $\psi$ such that $\sup_{x\in\Reals} |\psi'(x)| < \infty$. \vspace{-0.5em}
\begin{enumerate}\itemsep0em
\item $\displaystyle \int_{-\infty}^{+\infty} K(x)\, dx >0$ and for any positive integer $m$, $\displaystyle \int_{-\infty}^{+\infty} |x|^m\,|K(x)|\, dx <\infty$.
\item The function $K$ can be extended to a uniformly bounded analytic function on  the strip $\mathcal{S}(\rho) = \big\{z=x+\sqrt{-1}\, y\in \mathbb{C}:\,(x,\,y)\in\mathbb{R}^2,\, |y|\leq \rho\big\}$ in the complex plane for some constant $\rho>0$.
\end{enumerate}


Recall that $\mu_{t\ell}$ is the value assigned to leaf $\ell$ of tree $t$, for $\ell=1,2,\ldots,L_t$ and $t=1,\ldots, T$, and $\tau_b$ is the bandwidth parameter associated with branch $b$. Our first result shows that any smooth function can be approximated by a sum of soft decision trees taking form~\eqref{Eqn:SumOfTree} in a way such that the number of trees $T$ and the approximation error are optimally balanced. This lemma is interesting in its own right since it indicates that any $d$-dimensional smooth function can be approximated within error $\varepsilon$ by using at most poly$(\varepsilon^{-1})$ many properly re-scaled logistic activation functions.

\begin{lemma}[Approximation by sum of soft decision trees]
  \label{lem:half-gaus}
  Suppose Assumption G holds for the gating function $\psi$. For any function
  $f_0 \in \sC^{\alpha,R}([0,1]^d)$, any $\epsilon > 0$, and $\tau > 0$, there
  exists a sum of soft decision trees with a single bandwidth $\tau_b\equiv
  \tau$ for all branches,
  \begin{align*}
    \widetilde f(x) =  \sum_{t = 1}^T g(x;  \widetilde \Tree_t,  \widetilde \sM_t), \qquad x \in \Reals^p,
  \end{align*}
  where each tree $\widetilde{\Tree}_t$ has at most $2d$ branches, $T \le C_1\tau^{-d}\log^d(1/\epsilon)$, $\sum_{t,\ell} |\widetilde \mu_{t\ell}| \le C_1\,\tau^{-d}\,\|f_0\|_\infty$, and 
  \begin{align*}
    \| \widetilde f - f_0\|_\infty \le D_1\, R\, (\tau^\alpha + \varepsilon \, \tau^{-d}),
  \end{align*}
  where $C_1$ and $D_1$ are constants independent of $(\varepsilon, \tau)$.
  
\end{lemma}

With the help of this lemma, we establish \eqref{eq:lower-bound} as a direct consequence of the following result, where we make the following assumptions on the prior distribution.


\paragraph{{\bf Assumption P} (prior conditions):} 
\begin{enumerate}
\item[(P1)] There exists some constants
  $(C_1,C_2)$ such that the prior distribution on number of trees $T$ satisfies
  \begin{math}
    \Pi(T=t) \geq C_1\,\exp\{-C_2\, t\} ~ \mbox{for $t=0,1,2,\ldots$}.
  \end{math}
\item[(P2)] The prior density $\pi_\tau$ of tree specific bandwidth parameters $\tau_t$ satisfies $\pi_\tau(\tau)\geq a_1\tau^{a_2}$ for some constants $a_1, a_2 > 0$ for all sufficiently small $\tau$.
\item[(P3)] The prior on the splitting proportion vector $s$ is $\Dirichlet(a/p^{\xi}, \ldots, a/p^{\xi})$ for some $\xi > 1$ and $a>0$.
\item[(P4)] The leaf coefficients \(\mu_{t\ell}\) are iid with density $\pi_\mu$ where $\pi_\mu(\mu) \ge B_1 e^{-B_2 |\mu|}$ for all $\mu$ and some positive constants $B_1, B_2$.
\item[(P5)] $\Pi(D_t = k) > 0$ for $k=0,1,\ldots,2d$, where \(D_t\) denotes the depth of tree $t$ and $d$ is as in Theorem~\ref{thm:prior_thick}.
\end{enumerate}

\begin{remark}
  Condition P1 is very weak and is satisfied, for example, by setting $T \sim \Geometric(\pi_T)$. Similarly, P2 is satisfied by our choice of \(\tau_t \sim \Exponential(r)\). Condition P4, which assumes that the \(\mu_{t\ell}\)'s have sufficiently heavy tails, is adopted for the simplicity of the iid assumption, but can be weakened to allow for the hierarchical model in which $\mu_{t\ell} \sim \Normal(0,\sigma^2_\mu / T)$ with $\sigma_{\mu} \sim \Cauchy_+(0,\sigma_\sigma)$.
\end{remark}

\begin{remark}
  In the supplementary material we show that under extra technical conditions on the prior, the usual posterior (fractional posterior with $\eta=1$) can attain the same rate of convergence as in Theorem~\ref{Thm:main} below. These extra conditions are needed to control the size of the effective support of the prior and show the existence of a certain sieve \citep{ghosal2000convergence}. In particular, Assumption P only needs certain lower bounds on the prior density (mass) functions, while Assumption SP in the supplementary material requires some upper bound on the tail prior probability of various parameters in the model.
\end{remark}


\begin{theorem}(Prior concentration for sparse function)\label{thm:prior_thick}
  Suppose that Assumptions G and P are satisfied. Let $f_0\in\mathcal C^{\alpha,
    R}([0,1]^p)$ be a bounded regression function that depends on at most $d$
  covariates. Then there exists constants $A$ and $C$ independent of $(n, p)$
  such that for all sufficiently large $n$, the prior $\Pi$ over regression
  function $f$ satisfies
\begin{align*}
\Pi\Big[ \|f-f_0\|_\infty \leq A\, \varepsilon_n\Big] \geq \exp\big(-Cn\varepsilon_n^2\big),
\end{align*}
where $\varepsilon_n =n^{-\alpha/(2\alpha+d)}(\log n)^t + \sqrt{n^{-1} \,d\, \log p}$ for any $t\geq \alpha(d+1)/(2\alpha+d)$. 
\end{theorem}

The following posterior concentration rate for sparse functions follows immediately from Theorem~\ref{thm:prior_thick} and Theorem 3.2 in \cite{bhattacharya2016bayesian} (see also Section 4.1 therein).

\begin{theorem}[Posterior convergence rate for sparse truth]\label{Thm:main}
  Suppose that Assumptions G and P are satisfied. Let $f_0\in\mathcal C^{\alpha,
    R}([0,1]^p)$ be a bounded regression function that only depends on at most
  $d$ covariates. If $n\,\varepsilon_n^2\to\infty$ and $\varepsilon_n\to 0$ as $n,\,p\to\infty$, then for all sufficiently large
  constant $M>0$, we have
\begin{align*}
  \Pi_{n,\eta}\Big[ \|f-f_0\|_n\geq M\, \varepsilon_n\Big] \to 0,
  \quad \text{in probability as $n,\,p \to \infty$},
\end{align*}
where $\varepsilon_n =n^{-\alpha/(2\alpha+d)}(\log n)^t + \sqrt{n^{-1} \,d\, \log p}$ for any $t\geq \alpha(d+1)/(2\alpha+d)$. 
\end{theorem}

This result shows a salient feature of our sum of soft decision trees model --- by introducing the soft thresholding, the resulting posterior contraction rate adapts to the unknown smoothness level $\alpha$ of the truth $f_0$, attaining a near-minimax rate \citep{yang2015minimax} without the need of knowing $\alpha$ in advance. Our next result shows that if the truth admits a sparse additive structure $f_0=\sum_{v=1}^V f_{0,v}(x)$, where each additive component $f_{0,v}(x)$ is sparse and only depends on $d_v$ covariates for $v=1,\ldots, V$, then the posterior contraction rate also adaptively (with respect to both the additive structure and unknown smoothness of each additive component) attains a near-minimax rate \citep{yang2015minimax} up to $\log n$ terms, which leads to a second salient feature of the sum of soft decision tree model --- it also adaptively learns any unknown lower order nonlinear interactions among the covariates.

\begin{theorem}[Posterior convergence rate for additive sparse truth]\label{Thm:main2}
Suppose that Assumptions G and P are satisfied. Let $f_0=\sum_{v=1}^V f_{0,v}$, where the $v$th additive component $f_{0,v}$ belongs to $\mathcal C^{\alpha_v, R}([0,1]^p)$, and is bounded and only depends on at most $d_v$ covariates for $v=1,\ldots,V$. If $n\,\varepsilon_n^2\to\infty$ and $\varepsilon_n\to 0$ as $n,\,p\to\infty$, then for all sufficiently large constant $M>0$, we have
\begin{align*}
  \Pi_{n,\eta}\Big[ \|f-f_0\|_n\geq M\, \varepsilon_n\Big] \to 0,
  \quad \mbox{in probability} \quad \mbox{as $n,\,p\to\infty$},
\end{align*}
where $\varepsilon_n =\sum_{v=1}^Vn^{-\alpha_v/(2\alpha_v+d_v)}(\log n)^t + \sum_{v=1}^V \sqrt{n^{-1} \,d_v\, \log p}$ for any $t\geq \max_v \alpha_v(d_v+1)/(2\alpha_v+d_v)$.
\end{theorem}

\section{Illustrations}
\label{sec:illustrations}

\subsection{Friedman's example}
\label{sec:friedman}

A standard test case, initially proposed by \citet{friedman1991multivariate}
\citep[see also][]{chipman2010bart},
sets
\begin{align}
  \label{eq:friedfun}
  f_0(x) = 10 \sin(\pi x_1 x_2) + 20 (x_3 - 0.5)^2 + 10 x_4 + 5 x_5.
\end{align}
This $f_0(x)$ features two nonlinear terms, two linear terms, with a nonlinear interaction.

In this experiment, we consider \(n = 250\) observations, \(\sigma^2 \in \{1, 10\}\), and \(p\) from \(5\) to \(1000\) along an evenly-spaced grid on the scale of \(\log p\). We compare SBART to BART, DART, gradient boosted decision trees (\texttt{xgboost}), the lasso (\texttt{glmnet}), and random forests (\texttt{randomForest}). A similar experiment was conducted by \citet{linero2016bayesian}, who showed that the sparsity inducing prior used by DART resulted in substantial performance gains over BART. The purpose of this experiment is to demonstrate the further gains which are possible when the smoothness of \eqref{eq:friedfun} is also leveraged.


\begin{figure}[t]
  \centering
  \includegraphics[width = .77\textwidth]{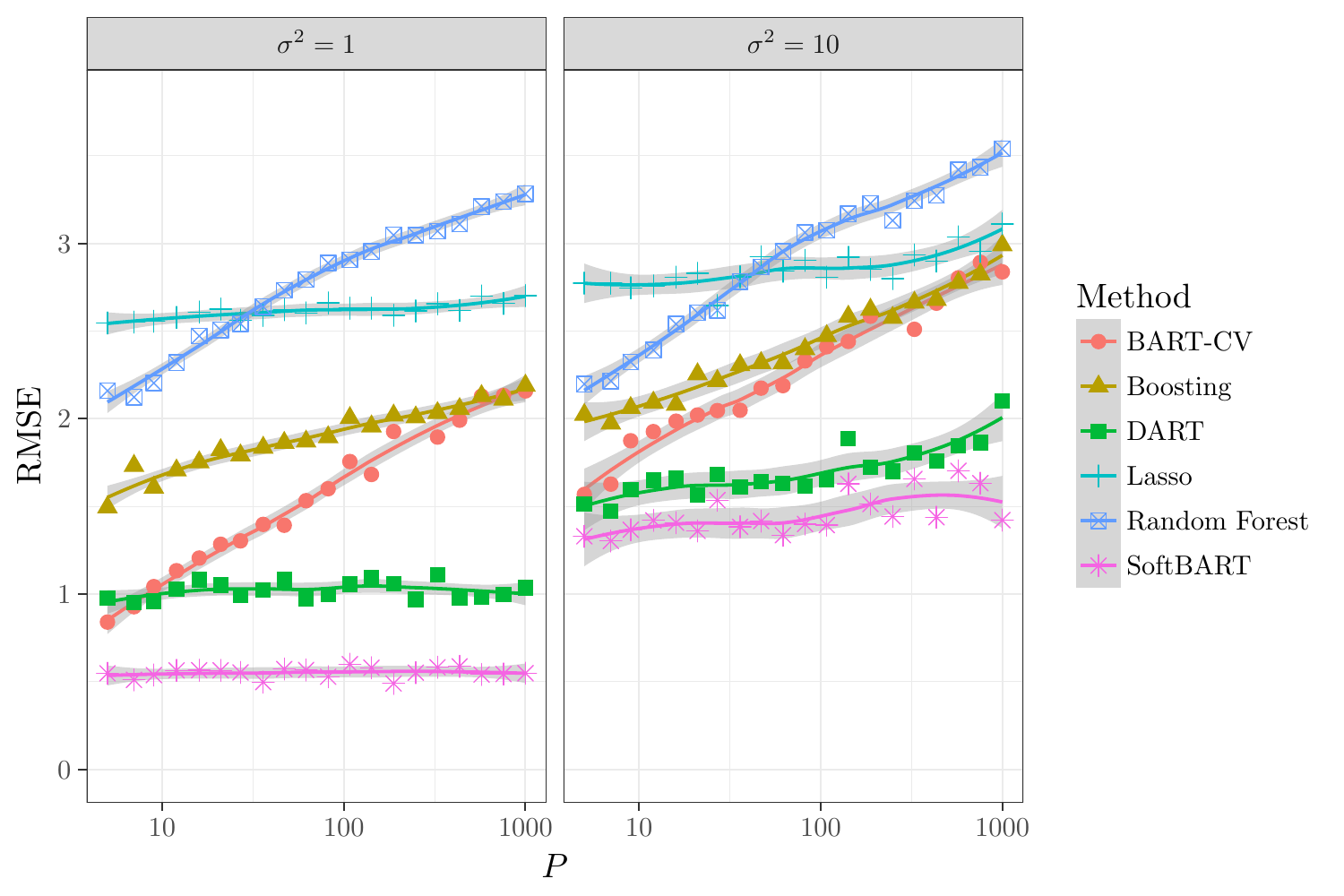}
  \caption{Average root mean squared error of various methods, as a function of
    the dimension \(P\) of the predictor space. To aide visualization, we also
    give a loess smooth with Monte-Carlo standard error.
  }
  \label{fig:pgrow}
\end{figure}

Methods are compared by root mean-squared error,
\begin{math}
\operatorname{RMSE} = \{\int \{f(x) - \widehat f(x)\}^2 \ dx\} ^ {1/2},
\end{math}
which is approximated by Monte-Carlo integration. For the Bayesian procedures, we take \(\widehat f\) to be the pointwise posterior mean of \(f\). DART, and SBART use their respective default priors and were fit using 2500 warmup iterations and 2500 sampling iterations, while cross-validation is used to tune the hyperparameters for BART. The non-Bayesian methods were tuned using cross validation for each replication of the experiment.

Results are given in Figure~\ref{fig:pgrow}. Among the methods considered, SBART performs the best, obtaining a sizeable improvement over DART in both the low noise and high noise settings. Due to the use of a sparsity-inducing prior, both DART and SBART are largely invariant to the number of nuisance predictors, while random forests, BART-CV, and boosting have errors increasing in \(\log p\). The lasso also has stable, albeit poor, performance as \(p\) increases.

We now compare SBART to DART for the task of variable selection (see \citealp{linero2016bayesian} for a detailed comparison of DART, BART, random forests, and the lasso which found DART to perform best among these methods). Our goal is to assess whether leveraging smoothness can improve on the good variable selection properties of DART. We modify Friedman's function, taking instead
\begin{align*}
  f(x) = 10 \sin(\pi x_1 x_2) + 20(x_3 - 0.5)^2 + \lambda(10 x_4 + 5 x_5),
\end{align*}
where $\lambda$ is a tuning parameter for the simulation. A variable is included if its posterior inclusion probability exceeds 50\%. We consider $\lambda \in [0.1, 1]$. As measures of accuracy, we consider precision $ = TP / (TP + FP)$, recall $= TP / (TP + FN)$, and $F_1$ score (harmonic mean of precision and recall), where $TP, FP$ and $FN$ denote the number of true positives, false positives, and false negatives respectively.

Results for $20$ replications and $\sigma^2 = 1$ are given in Figure~\ref{fig:varselect}, along with the average RMSE. First, we see that both DART and SBART have a precision which is roughly constant in $\lambda$, with SBART performing uniformly better. This makes intuitive sense, as varying $\lambda$ should have little influence on whether irrelevant predictors are selected. The precision of both methods is heavily dependent on $\lambda$, and we see that SBART is generally capable of detecting smaller signal levels; at its largest, the difference in recall is about 10\%. Once the signal level is high enough, both methods detect all relevant predictors consistently. The $F_1$ score reflects a mixture of these two behaviors. Perhaps most interesting is the influence of $\lambda$ on the RMSE. As $\lambda$ increases the performance of DART degrades while SBART remains roughly constant. Intuitively this is because, as $\lambda$ increases, DART must use an increasing number of branches to capture the additional signal in the data, while SBART is capable of representing the effects corresponding to $(x_4, x_5)$ with fewer parameters.

\begin{figure}
  \centering
  \includegraphics[width = .9\textwidth]{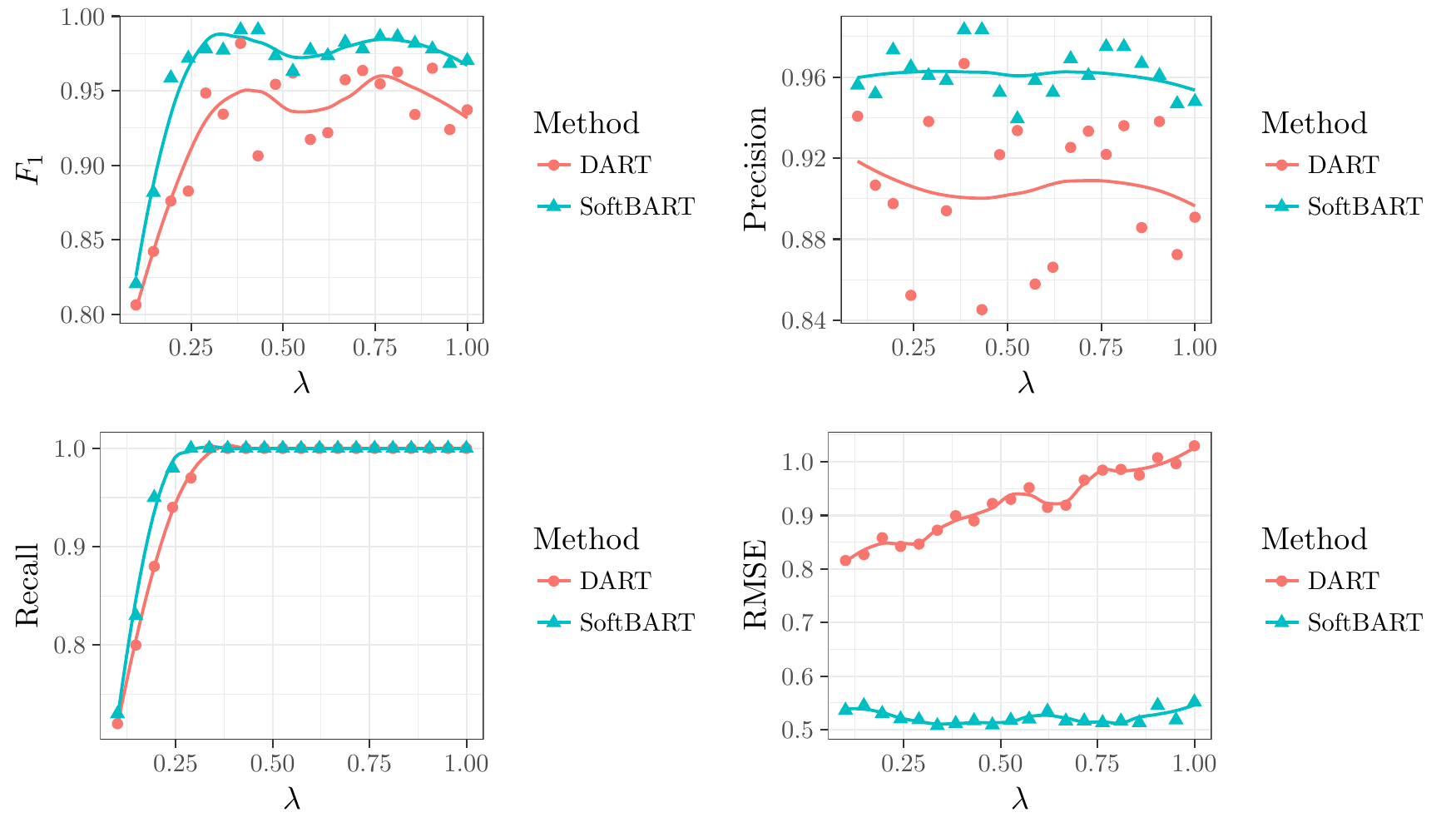}
  \caption{Results for variable selection, with a loess smooth to aide
    visualization.}
  \label{fig:varselect}
\end{figure}

\subsection{Approximation of non-smooth and locally smooth functions}
\label{sec:local}

A potential concern with the use of soft decision trees is that they may not be able to capture fine-scale variability in the underlying regression function. An extreme example of this is when $f$ is a step function. We consider the regression function
\begin{math}
  f(x) = 2 -4 I(x_1 < 0.5).
\end{math}
In this case, one might expect soft decision trees to perform suboptimally relative to hard decision trees because a soft decision tree must model the jump at $0$ in a continuous fashion.

Surprisingly, ensembles of soft decision trees can outperform ensembles of hard decision trees even in this case. Figure~\ref{fig:step-fun} shows fits of BART and SBART to $n = 250$ data points and a high signal of $\sigma = 0.1$. We see that both methods can capture the large jump discontinuity at $x_1 = 0.5$. SBART performs better away from the discontinuity, however, because the level of smoothness is allowed to vary at different points in the covariate space. The trees responsible for the jump discontinuity have small $\tau_t$'s to effectively replicate a step function, while elsewhere the trees have large $\tau_t$'s to allow the function to essentially be constant.

\begin{figure}
  \centering
  \includegraphics[width=.7\textwidth]{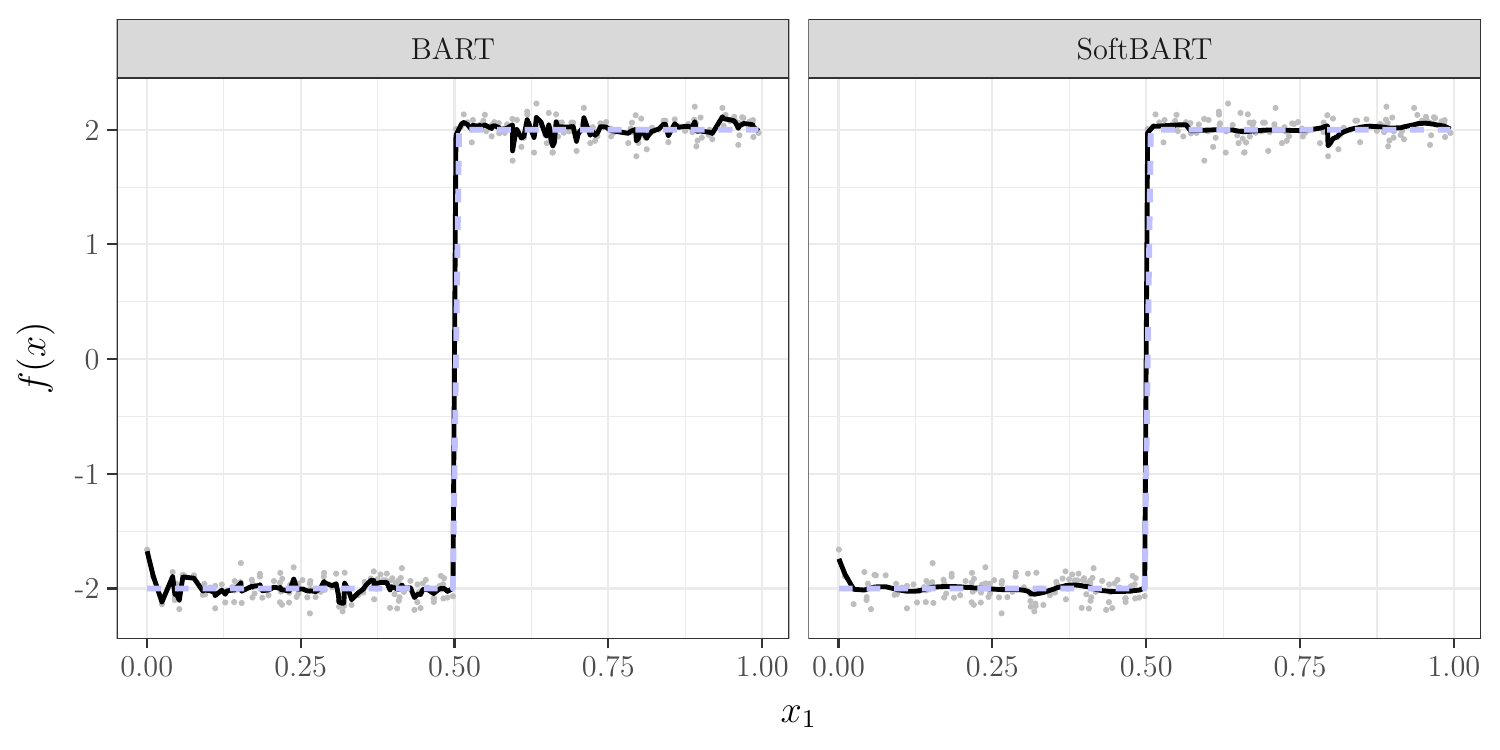}
  \caption{Estimate of $f(x) = 2 - 4I(x_1 < 0.5)$ using the posterior mean under
    the BART and SBART priors; dashed/blue line is the true mean, solid/dark
    line is the fit. Points are the observed data.}
  \label{fig:step-fun}
\end{figure}

The ability to select different $\tau_t$'s allows SBART to obtain a locally-adaptive behavior. 
To illustrate this, Figure~\ref{fig:wavelet} gives the fit of BART and SBART when $f(x)$ is a highly localised Daubechies wavelet of smoothness order $10$. We see that SBART is capable of adapting both to the constant regions outside of the support of the wavelet, and the fast oscillatory behavior within the support of the wavelet. The fit of BART, by contrast, possesses many artifacts outside the support of the wavelet, and possesses generally wider credible bands. 

\begin{figure}
  \centering
  \includegraphics[width=.9\textwidth]{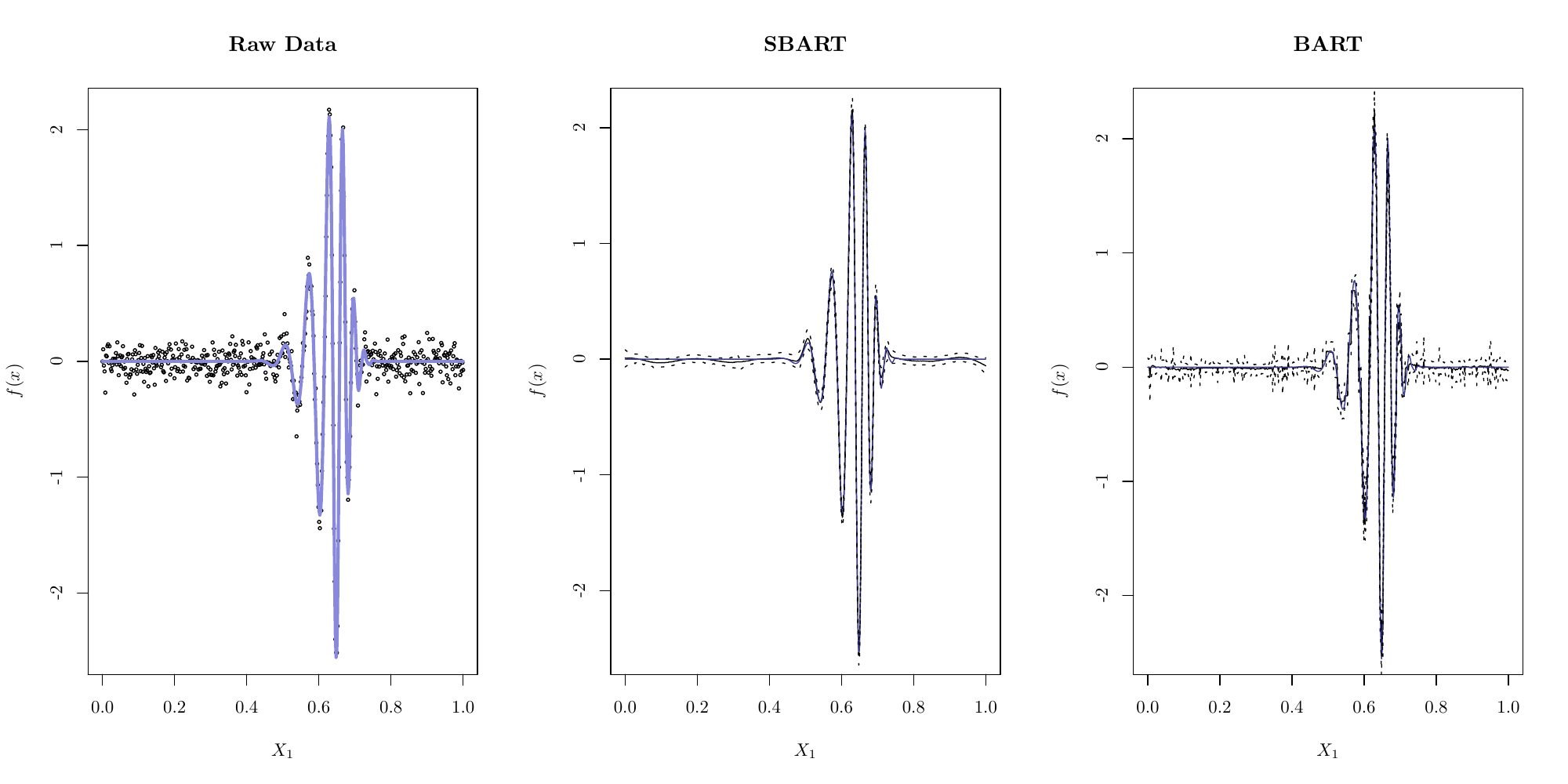}
  \caption{Left: Raw data, consisting of observations drawn with the Daubechies wavelet as the mean function, with the mean function displayed in light blue. Middle: Fit of SBART to the raw data, with pointwise 95\% posterior credible bands. Right: Fit of BART to the raw data, with pointwise 95\% posterior credible bands.}
  \label{fig:wavelet}
\end{figure}

\subsection{Benchmark datasets}
\label{sec:benchmark}

We compare the SBART to various tree-based and non-tree-based methods on several benchmark datasets. We consider BART, DART, the LASSO (\texttt{glmnet}), random forests (\texttt{randomForest}), and gradient boosted decision trees (\texttt{xgboost}). The parameters for the non-Bayesian procedures were chosen, separately for each fit, using the \texttt{caret} package. Default priors (with $T = 50$) for SBART and DART were used; additionally, we consider selecting the hyperparameters of SBART and BART by cross validation.



Ten datasets are considered. Aside from \texttt{bbb} and \texttt{wipp}, the datasets are a subset of those considered by \citet{kim2007visualizable}. While we consider only a subset of these datasets, no datasets considered for this experiment were omitted. Attributes of these datasets are presented in Table~\ref{tab:datasets}. The response in each dataset was transformed to be approximately Gaussian. The \texttt{bbb}, \texttt{triazines}, and \texttt{wipp} datasets were also considered by \citet{linero2016bayesian} to illustrate features of the sparsity-inducing priors for decision tree methods.

\begin{table}
\begin{tabular}{lrrrrrrr}
\toprule 
Data & BART-CV & DART & SBART & RF & XGB & LASSO & SBART-CV\\
\midrule
\texttt{ais} & \textbf{1.00} (1) & \textbf{1.00} (1) & \textbf{1.00} (1) & 1.01 (5) & 1.03 (6) & 1.04 (7) & \textbf{1.00} (1)\\
\texttt{abalone} & 1.03 (4) & 1.03 (4) & \textbf{1.00} (1) & 1.02 (3) & 1.03 (4) & 1.12 (7) & \textbf{1.00} (1)\\
\texttt{bbb} & 1.07 (6) & 1.04 (4) & \textbf{0.99} (1) & 1.01 (3) & 1.05 (5) & 1.10 (7) & 1.00 (2)\\
\texttt{cpu} & 0.98 (2) & 1.01 (5) & 1.01 (4) & \textbf{0.97} (1) & 1.02 (6) & 1.31 (7) & 1.00 (3)\\
\texttt{diamonds} & 1.15 (4) & 1.07 (3) & 1.01 (2) & 2.29 (6) & 1.43 (5) & 3.53 (7) & \textbf{1.00} (1)\\
\texttt{hatco} & 1.14 (3) & 1.15 (4) & 1.10 (2) & 1.39 (6) & 1.20 (5) & 1.44 (7) & \textbf{1.00} (1)\\
\texttt{servo} & 1.02 (3) & 1.02 (3) & \textbf{0.99} (1) & 1.17 (6) & 1.06 (5) & 1.75 (7) & 1.00 (2)\\
\texttt{tecator} & 1.87 (4) & 1.63 (4) & \textbf{0.98} (1) & 1.95 (7) & 1.56 (3) & 1.85 (5) & 1.00 (2)\\
\texttt{triazines} & 0.98 (3) & 0.99 (4) & 0.99 (4) & \textbf{0.92} (1) & 0.94 (2) & 1.13 (7) & 1.00 (6)\\
\texttt{wipp} & 1.19 (4) & 1.14 (3) & 1.03 (2) & 1.43 (7) & 1.28 (5) & 1.41 (6) & \textbf{1.00} (1) \\
\midrule
Average RMPE & 1.14 (4) & 1.11 (3) & 1.01 (2) & 1.32 (6) & 1.16 (5) & 1.57 (7) & \textbf{1.00} (1)\\
Average Rank & 3.4 (3) & 3.5 (4) & \textbf{1.9} (1) & 4.5 (5) & 4.6 (6) & 6.7 (7) & 2 (2)\\
\bottomrule
\end{tabular}
\caption{Results of the experiment described in Section~\ref{sec:benchmark}. The columns associated with the methods give their root mean predictive error, normalized by the root mean predictive error of SBART-CV. In parentheses, we give the rank of the method among the 5 different approaches. The best-ranked method for each dataset is given in bold.}
  \label{tab:datasets}
\end{table}

Results of the experiment are given in Table~\ref{tab:datasets}. Methods are compared by an estimate of their root mean predictive error obtained using $5$-fold cross-validation, with the results averaged over $20$ replications of the cross-validation. For each experiment, the root mean predictive error for each method is normalized by the root mean predictive error for SBART-CV, so that scores higher than 1.00 correspond to worse performance than SBART and scores lower than 1.00 correspond to better performance than SBART.

SBART/SBART-CV is seen to perform very well in practice, attaining the best performance on 8 out of the 10 datasets. The results here are consistent with the general observation of \citet{chipman2010bart} that BART outperforms gradient boosting and random forests in aggregate over many datasets. Two datasets stand out as particularly interesting. First, for the \texttt{tecator} dataset, SBART outperforms all other methods by a very wide margin, indicating that leveraging smoothness for this dataset is essential to attaining good performance. Second, the only dataset for which SBART-CV substantially outperforms SBART is the \texttt{hatco} dataset, where tuning the number of trees is required to attain optimal performance. This indicates that, for most datasets, the default SBART procedure works very well, but that if one wants to be absolutely sure of optimal performance they should tune $T$. 


\section{Discussion}
\label{sec:discussion}

We have introduced a novel Bayesian sum-of-trees framework and demonstrated that
it is capable of attaining a meaningful improvement over existing methods both
in simulated experiments and in practice. This was accomplished by incorporating
soft decision trees and sparsity-inducing priors. We also provided theoretical
support in the form of near-optimal results for posterior concentration,
adaptively over smoothness classes, when $f_0(x)$ is a sparse, or additive,
function. 

While this paper has focused only on the case of nonparametric regression, the
proposed methodology extends in a straight-forward manner to other settings. For
example, the case of binary classification can be addressed in the usual way via
a probit link and data augmentation.

Our theoretical results concern the rate of convergence of the posterior.
Another relevant question is whether the model can consistently estimate the
model support. That is, one can ask under what conditions $\Pi(S = S_0 \mid
\Data) \to 1$ as $n \to \infty$, where $S = \{p : \text{predictor $p$ appears in
  the ensemble}\}$ and $S_0 = \{p: \text{$f_0$ depends on $p$}\}$. This is an
interesting area for future research.

Software which implements SBART is available online at \url{https://github.com/theodds/SoftBART}, and is undergoing active development. Our code is based on the implementation of BART in the \texttt{BayesTree} package. Given enough optimization, we hope that our implementation could reach speeds within a modest factor of existing highly-optimized implementations of BART \citep{kapelner2014bartmachine}.

\appendix

\section{Proof of Lemma~\ref{lem:half-gaus}}

Let $K^{(d)}_{\tau}(x_1,x_2,\ldots,x_d)=\tau^{-d}\,\prod_{j=1}^d K(x_j/\tau)$ denote a $d$-dimensional tensor product of the rescaled one dimensional kernel $K$ in Assumption G, where recall that $\tau$ is the bandwidth parameter in the gating function. Let $\displaystyle C_K:\,=\int K(x)\,dx$ denote the normalization constant of $K$, so that we can write $K=C_K\widetilde K$, and the rescaled kernel function $\widetilde K$ is has an unit normalization constant. Also write $\widetilde K^{(d)}_\tau = K^{(d)}_{\tau}/C_K^d$. It is easy to verify that $\widetilde K$ also satisfies the two conditions in Assumption G, \hlb{though it may not be associated to any $\widetilde \psi$}.

Our proof is composed of three steps. First, we provide error bound estimates of
approximating any $\alpha$-smooth function by a convolution $K^{(d)}_{\tau} * g$ with some carefully constructed function $g$ for any $\tau > 0$. Second, we show that any continuous convolution $K^{(d)}_{\tau} * g$ can be approximated by a discrete sum $\sum_{t = 1}^T
\mu_t K^{(d)}_{\tau}(\cdot - x_t)$ with at most $O(\tau^{-d})$ atoms.
Lastly, we provide an error bound estimate on approximating this sum of kernels with a sum
of soft decision trees by identifying each kernel component $K^{(d)}_{\tau}(\cdot - x_t)$ as one particular leaf in the $t$th soft decision tree $g(x;\, \mathcal{T}_t,\mathcal{M}_t)$ whose depth is at most $2d$ via splitting at most $2d$ times, for $t=1,\ldots,T$.

{\bf Step 1:} This step is follows as a direct result of the following lemma, which is
adapted from Lemma 3.4 of \cite{de2010adaptive}. 

\begin{lemma}
  Under Assumption G, for any $f_0 \in \sC^{\alpha,R}([0,1]^d)$, there exist some constants $(M_1,M_2)$ independent of $\tau$, and a function $T_{b,\tau}
  f_0$ satisfying $\|T_{b,\tau}
  f_0\|_\infty \leq M_1R$, such
  that
  \begin{align*}
    \|\widetilde K^{(d)}_{\tau} * (T_{b, \tau}f_0) - f_0 \|_{\infty} \le M_2 R\, \tau^\alpha.
  \end{align*}
\end{lemma}

From this lemma, we immediately have 
  \begin{align*}
    \| K^{(d)}_{\tau} * g - f_0 \|_{\infty} \le M_2 R\, \tau^\alpha,
  \end{align*}
where $g = C_K^{-d}\,T_{b, \tau} f_0$ satisfies $\|g\|_\infty \leq M_1'R$, with $M_1'=C_K^{-d}M_1$ independent of $\tau$.

{\bf Step 2:} This step generalizes the theory of approximating a continuous one-dimensional
density function from by a mixture of Gaussians developed in
\citet{ghosal2007posterior} to by a location mixture of any kernel $K$ satisfying Assumption G. We also extend their result from density estimation to general function estimation as demanded in our regression setting,
where the target function $f$ may not integrate to one and can take
negative values. First, we state an extension of Lemma 3.1 of
\citet{ghosal2001entropies} from dimension one to dimension $d$, and from the Gaussian kernel to any kernel $K$ satisfying Assumption G.

\begin{lemma}
  \label{lem:mix-approx}
  Under Assumption G, for any probability density function $p_0$ on $[0,1]^d$, any $\epsilon > 0$,
  and $\tau \in (0,1)$, there is a discrete measure $P_\tau = \sum_{t = 1}^T r_t\,
  \delta_{x_t}$ with $T \le C_1 \tau^{-d} \log^d(1/\tau)$ support points such
  that $\sum_{t = 1}^T r_t = 1$ and
  \begin{align*}
    \|K^{(d)}_{\tau} * p_0 - \sum_{t=1}^T r_t \, K^{(d)}_{\tau}(\cdot - x_t)\|_\infty
    \le
    D_1 \epsilon / \tau^d,
  \end{align*}
  where $(C_1, D_1)$ are independent of $\tau$ and $K$. 
\end{lemma}

\begin{proof}
We only sketch the key difference in the proof from Lemma 3.1 of
\citet{ghosal2001entropies} in the one-dimensional case, and a proof for extending the result from one-dimensional case to the multi-dimensional case follows similar lines as in the proof of Theorem 7 in \cite{shen2013adaptive} (by replacing the Gaussian kernel with the kernel $K$).

The only key property of the Gaussian kernel used in the proof of Lemma 3.1 of \citet{ghosal2001entropies} is in bounding the remainder term in the $k$-th order Taylor expansion in their equation (3.11),
where they used the fact that for any $k\geq 1$, the $k$th order derivative of the standard Gaussian density function $\phi(x)=(2\pi)^{-1/2}\,e^{-x^2/2}$ at the origin $x=0$ satisfies the bound
$$
\Big|\frac{\phi^{(k)}(0)}{k!}\Big| \leq C_1\,\exp\{-C_2\, k\}, 
$$
for some sufficiently large constant $C_2>0$ (since we only focus on the approximation error over the unit interval $[0,1]$, we do not need to include the additional $\log(1/\varepsilon)$ term in equation (3.11) therein). Therefore, it suffices to verify a similar exponentially decay bound for the $k$th order derivative of function $K_\kappa:\,=\tau^{-1}\,K(\cdot/\kappa)$ for some sufficiently large number $\kappa>0$ depending on $C$. In fact, under Assumption G, $K(\cdot)$ can be analytically extend to the strip $\mathcal S(\rho)$ in the complex plane (for simplicity, we use the same notation $K$ to denote this extension), which implies by applying Cauchy's integral formula that 
\begin{align*}
\frac{K_\kappa^{(k)}(0)}{k!} = \frac{1}{2\pi \sqrt{-1}}\oint_{\Gamma_\kappa} \frac{K_\kappa(z)}{z^{k+1}}\,dz,
\end{align*}
where the closed path $\Gamma_\kappa$ is chosen as a counter-clockwise circle centering at the origin with radius $\kappa\,\rho$. Since $K$ is uniformly bounded on the path $\Gamma_\kappa$ by Assumption G, we can further deduce that
\begin{align*}
\Big|\frac{K_\kappa^{(k)}(0)}{k!}\Big| \leq \frac{D}{\kappa^{k+2}\,\rho^{k+1}} \leq  D\,\exp\{-C_2\, k\}
\end{align*}
holds as long as $\kappa \geq\rho^{-1}\,\exp\{C_2\}$, where $D$ is some constant only depending on $K$, which completes the proof.
\end{proof}

With this lemma on the density function approximation as our preparation, we now return to the problem of approximating any general bounded function $g$ over $[0,1]^d$. Notice that we always have the decomposition $g=g_{+}-g_{-}$ where $g_{+}=\max\{0, g(x)\}$ and $g_{-}(x) =\max\{0, -g(x)\}$ are the positive parts and negative parts of $g$, respectively, and both of them are nonnegative and bounded over $[0,1]^d$. Let $A_+=\int_{[0,1]^d} g_+(x)\,dx \leq \|g\|_{\infty}$ and$A_-=\int_{[0,1]^d} g_-(x)\,dx\leq \|g\|_{\infty}$. 
It is obvious that $g_{+}/A_{+}$ and $g_{-}/A_{-}$ are two legitimate pdfs over $[0,1]^d$. By applying Lemma~\ref{lem:mix-approx}, we can find two discrete measures $P_{+}=\sum_{t=1}^{T_+} r^{+}_t\, \delta_{x^{+}_t}$ and $P_{-}=\sum_{t=1}^{T_-} r^{-}_t\, \delta_{x^{-}_t}$ such that
\begin{align*}
&\big|A_{+}^{-1} K^{(d)}_{\tau}*g_+(x) -  \sum_{t=1}^{T_+} r^+_t \,K^{(d)}_{\tau}(x - x^+_t)\big|\leq D\,\varepsilon/\tau^d,\\
&\big|A_{-}^{-1} K^{(d)}_{\tau}*g_-(x) -  \sum_{t=1}^{T_-} r^-_t\, K^{(d)}_{\tau}(x - x^-_t)\big|\leq D\,\varepsilon/\tau^d,
\end{align*}
for any $x\in[0,1]^d$ and $\max\{T_+,\,T_-\} \leq  C\, \tau^{-d} \log^d(1/\varepsilon)$.
Now we combine these two discrete measures into a new discrete signed measure $P_0=\sum_{t=1}^{T_+} A_{+}\, r^+_t \,K^{(d)}_{\tau}(x - x^+_t) + \sum_{t=1}^{T_-} (-A_{-}\, r^-_t )\,K^{(d)}_{\tau}(x - x^-_t)$, which will be denoted as $\sum_{t=1}^T \mu_t\, K^{(d)}_{\tau}\big(\cdot-x_t)$. Then $T\leq T_{-} + T_+ \leq 2C\, \tau^{-d} \log^d(1/\varepsilon)$ and 
\begin{align*}
\big|K^{(d)}_{\tau}*g(x) -  \sum_{t=1}^{T} \mu_t \,K^{(d)}_{\tau}(x - x_t)\big|\leq (A_++A_-)\, D\,\varepsilon/\tau^d\leq 2D\,\|g\|_\infty \, \varepsilon/\tau^d,
\end{align*}
for all $x\in[0,1]^d$. Moreover, we have $\sum_{t=1}^T|\mu_t| \leq A_+\sum_{t=1}^{T_+} r_t^+ + A_-\sum_{t=1}^{T_-} r_t^- \leq 2\|g\|_\infty$.

{\bf Step 3:} In the last step, for each component $\mu_t\, K^{(d)}_{\tau}\big(\cdot-x_t)$ in the sum, we construct a soft decision tree $\widetilde{\mathcal{T}}_t$ and its associated leaf values $\widetilde{\mathcal{M}}_t$ in a way such that: 1. the tree splits exactly $2d$ times; 2. the weight function $\phi(x;\,\mathcal{T}_t,\, \ell_t)$ specified in~\eqref{eq:proper} associated with one particular leaf $\ell_t$ equals to $\tau^d\, K^{(d)}_{\tau}\big(\cdot-x_t)$, so that the existence of the sum of soft decision tree follows by setting the values $\widetilde {\mu}_{tl}$ associated with other leaves $\ell\neq \ell_t$ in this tree to be zero, and the value of this leave as $\widetilde {\mu}_{tl_t}=\tau^{-d}\, \mu_t$.
In fact, for any $y=(y_1,\ldots,y_d)\in[0,1]^d$, we have the decomposition $K^{(d)}_{\tau}(y) = \prod_{j=1}^d \tau^{-d}\, \psi(y_j/\tau)\,\big(1-\psi(y_j/\tau)\big)$. Consequently, we can construct the tree $\mathcal{T}_t$ by sequentially splitting twice along each coordinate $x_{t,j}$ ($j=1,2,\ldots,d$) of the center $x_{t}=(x_{t,1},\ldots,x_{t,d})$ in $\mu_t\, K^{(d)}_{\tau}\big(\cdot-x_t)$, so that the particular leaf as the end point of the path that goes once left and once right, respectively, at the two branches associated with $x_{t,j}$, for $j=1,\ldots,d$, receives weight $\phi(\cdot;\,\mathcal{T}_t,\, \ell_t)=\prod_{j=1}^d \psi\big((\cdot-x_{t,j})/\tau\big)\,\big\{1-\psi\big((\cdot-x_{t,j})/\tau\big)\big\}= \tau^d\, K^{(d)}_{\tau}(x_t)$, implying that for any $x$, $g(x;\,\widetilde{\mathcal{T}}_t,\widetilde {\mathcal{M}}_t) = \widetilde {\mu}_{tl_t}\, \phi(x;\,\widetilde{\mathcal{T}}_t,\, \ell_t) = \mu_t\, K^{(d)}_{\tau}\big(x-x_t)$. Since this construction is valid for any $t=1,\ldots,T$, we have $\sum_{t=1}^T \mu_t\, K^{(d)}_{\tau}\big(x-x_t) = \sum_{t=1}^Tg(x;\,\widetilde{\mathcal{T}}_t,\widetilde {\mathcal{M}}_t)$.

\vspace{0.5em}

Finally, a combination of steps 1-3 together yields a proof of the lemma.

\section{Proof of Theorem~\ref{thm:prior_thick}}
For convenience, we use the same notation $C$ to denote some constant independent of $(n,p)$, whose value may change from line to line. Without loss of generality, we may assume that $f_0$ depends only on its first $d$ coordinates. Applying Lemma~\ref{lem:half-gaus}, we obtain that for some parameters $\tau$ and $\varepsilon$ to be determined later, there exists some $\widetilde f=\sum_{t=1}^{\widetilde{T}}g(x;\,\widetilde{\mathcal{T}}_t,\widetilde {\mathcal{M}}_t)$ such that $\widetilde{T} \leq C\, \tau^{-d} \log^d(\varepsilon^{-1})$, $\|\widetilde f -f_0\|_\infty \leq C\,(\tau^{\alpha} + \varepsilon\,\tau^{-d})$, and the total number of splits (all are along the first $d$ coordinates) across all trees are at most $2d\,\widetilde{T}$ ($2d\,\widetilde{T}$ many leaves in total).

Recall that our prior over the sum of soft decision tree function $f$ is specified in a hierarchical manner: first, we specify the number $T$ of trees and the tree topology $\mathcal{T}=\{\mathcal{T}_1,\ldots,\mathcal{T}_T\}$; second, conditional on these we decide the coordinates in all splits across all the decision trees; third, we sample the independent splitting locations along all the selected coordinates; last, we sample bandwidth parameters $\tau_t$ associated with each tree and parameters $\mu$'s associated with all leaves across the trees. We denote by $\widetilde T$ and $\widetilde{\mathcal{T}}$ the corresponding number of trees and the tree topology of $\widetilde{f}$.

We denote all the splitting coordinates of $f$ given $T$ and the tree topology $\mathcal{T}$ by $S\in\{1,\ldots,p\}^{N}$, where $N=\sum_{t=1}^T (L_t-1) \leq 2d \,T$ and recall that $L_t$ denotes the number of leaves in the $t$th tree, and denote by $\widetilde{S}$ the corresponding vector associated with $\widetilde f$. We also denote the set of all splitting locations (along the selected splitting coordinates) and bandwidths as $x=(x_1,x_2,\ldots,x_N)\in\mathbb{R}^N$ and $\tau_S=(\tau_1,\tau_2,\ldots,\tau_T)\in\mathbb{R}_{+}^T$ respectively, and the set of all leaf values as $\mu=(\mu_1,\ldots,\mu_{N+T})\in\mathbb{R}^{N+T}$. We also define $\widetilde{x}^N$ and $\widetilde{\mu}$ in a similar way.
By construction, it is easy to check that if $f$ shares the same $T$, tree topology $\mathcal{T}$ and splitting coordinates $S$ as $\widetilde{f}$, then if $\{x,\tau_S,\mu\}$ are sufficiently close to $\{\widetilde{x},\tau,\widetilde{\mu}\}$ in the sense that for any $\delta>0$,
\begin{align*}
&\max_{u=1,\ldots,N}\big|x_u-\widetilde{x}_u\big| \leq C\,\tau^{2d}\, \delta,\quad  \max_{u=1,\ldots,T}\big|\tau_u-\tau\big| \leq C\, \tau^{d+1}\,\delta,\\
&\quad  \mbox{and}\quad   \max_{u=1,\ldots, N+T}\big|\mu_u-\widetilde{\mu}_{u}\big| \leq C\, T^{-1}\,\tau^{d} \,\delta,
\end{align*}
then we have the following perturbation error bound by applying the triangle inequality,
\begin{align}\label{Eqn:pur_bound}
\bigg|\sum_{t=1}^{T}g(x;\,\mathcal{T}_t,\mathcal{M}_t) - \sum_{t=1}^{\widetilde{T}}g(x;\,\widetilde{\mathcal{T}}_t,\widetilde {\mathcal{M}}_t)\bigg| \leq C\, \delta,\quad\mbox{for all $x\in[0,1]^p$.}
\end{align}

Now we apply Theorem 2.1 in \cite{yang2014minimax} on the prior concentration probability for high-dimensional Dirichlet distribution and Assumption P3 to obtain that the splitting proportion vector $s=(s_1,\ldots,s_p)$ satisfies
\begin{align}\label{Eqn:sprior}
\Pi\Big[s_j \geq (2d)^{-1} \mbox{ for } j=1,\ldots,d, \mbox{ and }\sum_{j=d+1}^p s_j \leq d^{-1}\Big] \geq \exp\{-C\, d\log p\}.
\end{align}
This combined with the fact that each tree has depth at most $2d$ and Assumption P5 implies that the prior probability of $\mathcal{T}=\widetilde{\mathcal{T}}$ given $T=\widetilde{T}$ can be lower bounded by
\begin{align*}
\Pi\Big[\mathcal{T}=\widetilde{\mathcal{T}}\,|\, T=\widetilde{T}\big] \geq C\, d^{-N} \geq \exp\big\{-C\, \tau^{-d} \,\log^d(\varepsilon^{-1})\big\},
\end{align*}
where we have used the fact that $N \leq C\, \tau^{-d} \,\log^d(\varepsilon^{-1})$ in the last step.
The perturbation error bound in~\eqref{Eqn:pur_bound} implies
\begin{align*}
&\Pi\big[ \|f-\widetilde{f}\|_\infty \leq C\,\delta\,\big| \, \mathcal{T}=\widetilde{\mathcal{T}},\, T=\widetilde{T}\big] \\
&\geq \Pi\Big[\max_{u=1,\ldots,N}\big|x_u-\widetilde{x}_u\big| \leq C\,\tau^{2d},\ \max_{u=1,\ldots,T}\big|\tau_u-\tau\big| \leq C\, \tau^{d+1}\,\delta,\\
&\qquad\qquad\qquad\qquad\qquad\qquad    \max_{u=1,\ldots, N+T}\big|\mu_u-\widetilde{\mu}_{u}\big| \leq C\, T^{-1}\,\tau^{d} \,\delta\,\Big|\, \mathcal{T}=\widetilde{\mathcal{T}},\, T=\widetilde{T}\Big] \\
&\geq
\exp\big\{-C\, \tau^{-d} \log^d(\varepsilon^{-1}) \log \big[(\tau\,\delta)^{-1}\big]\big\},
\end{align*}
where in the last step we applied Assumptions P2 and P4, and used the fact that $\sum_{u}|\widetilde{\mu}_u|\leq C\, \tau^{-d}$ for some constant $C$ only dependent of $f_0$ (due to Lemma~\ref{lem:half-gaus}).
Putting all pieces together and using Assumption P1 and the properties of $\widetilde{f}$, we obtain
\begin{align*}
&\Pi\big[ \|f-f_0\|_\infty \leq C\,(\delta + \tau^\alpha+\varepsilon/\tau^d)\big] \\
&\geq \Pi[T=\widetilde{T}] \cdot \Pi\Big[\mathcal{T}=\widetilde{\mathcal{T}}\,|\, T=\widetilde{T}\big] \cdot \Pi\big[ \|f-\widetilde{f}\|_\infty \leq C\,\delta\,\big| \, \mathcal{T}=\widetilde{\mathcal{T}},\, T=\widetilde{T}\big]\\
&\geq \exp\big\{-C\,\tau^{-d} \log^d(\varepsilon^{-1}) -C\, d\log p-  C\, \tau^{-d} \log^d(\varepsilon^{-1}) \log \big[(\tau\,\delta)^{-1}\big]\big\}.
\end{align*}
Therefore, by choosing $\tau = (\log^{d+1} n/n)^{-1/(2\alpha + d)}$, $\delta = \tau^\alpha$, $\varepsilon=\tau^{d+\alpha}$, we can obtain the claimed prior concentration probability lower bound as $\Pi\big[ \|f-f_0\|_\infty \leq C\,\varepsilon_n\big] \geq \exp\{-C\,n\,\varepsilon_n^2\}$.


\section{Proof of Theorem~\ref{Thm:main2}}
Using Theorem 3.2 in \cite{bhattacharya2016bayesian} (see also Section 4.1 therein), it suffices to show that $\Pi\big[ \|f-f_0\|_\infty \leq C\,\varepsilon_n\big] \geq \exp\{-C\,n\,\varepsilon_n^2\}$.
The proof of this is almost the same as that of Theorem~\ref{thm:prior_thick}, the only difference is that now we apply Lemma~\ref{lem:half-gaus} to find $V$ functions $\{\widetilde{f}_v:\,v=1,\ldots,V\}$, where  $\widetilde{f}_v$ contains $\widetilde{T}_v$ trees and approximates the $v$th additive component $f_{0,v}$ in $f_0$ for $v=1,\ldots, V$, and set $\widetilde{f}=\sum_{v=1}^V \widetilde{f}_v$. Due to the additive structure in our sum of soft decision tree model, we can always write $f=\sum_{v=1}^V f_v$ where $f_v$ collects $\widetilde{T}_v$ trees and has the same sum of soft decision tree prior structure when conditioning on the total number of trees $T=\sum_{v=1}^V\widetilde{T}_v$, and the conditional priors of $(f_1,\ldots,f_v)$ given $T=\sum_{v=1}^V\widetilde{T}_v$ and the splitting proportion vector $s$ are independent. Let $\mathcal S= \big\{s_j \geq (2d)^{-1} \mbox{ for } j=1,\ldots,d, \mbox{ and }\sum_{j=d+1}^p s_j \leq d^{-1}\big\}$ denote the event in inequality~\eqref{Eqn:sprior} with $d:\,=\sum_{v=1}^V d_v$.
Therefore, we obtain by applying Assumption P, the prior concentration bound~\eqref{Eqn:sprior} for $s$, and Theorem~\ref{thm:prior_thick} for a single $f_v$ (choose parameters $\tau_v,\delta_v,\varepsilon_v$ for each $f_v$ as in the proof of Theorem~\ref{thm:prior_thick}) that
\begin{align*}
&\Pi\Big[ \|f-f_0\|_\infty \leq C \sum_{v=1}^V \varepsilon_{n,v} \Big]\\
&\geq \Pi\Big[T = \sum_{v=1}^V\widetilde{T}_v \Big]\cdot \Pi\Big[s\in\mathcal{S}\Big] \cdot  \sup_{s\in\mathcal{S}}\bigg\{\prod_{v=1}^V\Pi\Big[ \|f_v-f_{0,v}\|_\infty \leq C \varepsilon_{n,v}\,\big|\, T=\sum_{v=1}^V\widetilde{T}_v,\,s\Big]\bigg\}\\
& \geq \exp\Big\{-C n \, \sum_{v=1}^V \varepsilon_{n,v}^2 - C\sum_{v=1}^V d_v \log p\Big\} \geq \exp\Big\{-C' n \, \Big(\sum_{v=1}^V \varepsilon_{n,v}\Big)^2\Big\}
\end{align*}
where constants $C, C'>0$, $\varepsilon_{n,v} = n^{-\alpha_v/(2\alpha_v+d_v)}(\log n)^{t_v} + \sqrt{n^{-1} \,d_v\, \log p}$ and $t_v\geq \alpha_v(d_v+1)/(2\alpha_v+d_v)$.

\section*{Acknowledgments}

This work was partially supported by NSF grant DMS-1712870 and DOD grant SOT-FSU-FATs-16-06.

\bibliographystyle{apalike}
\onehalfspacing
\bibliography{mybib.bib}

\begin{thebibliography}{}

\bibitem[Alaa and van~der Schaar, 2017]{alaa2017bayesian}
Alaa, A.~M. and van~der Schaar, M. (2017).
\newblock Bayesian nonparametric causal inference: Information rates and
  learning algorithms.
\newblock {\em arXiv preprint arXiv:1712.08914}.

\bibitem[Athreya and Ney, 2004]{athreya2004branching}
Athreya, K.~B. and Ney, P.~E. (2004).
\newblock {\em Branching processes}.
\newblock Courier Corporation.

\bibitem[Bhattacharya et~al., 2016]{bhattacharya2016bayesian}
Bhattacharya, A., Pati, D., and Yang, Y. (2016).
\newblock Bayesian fractional posteriors.
\newblock {\em arXiv preprint arXiv:1611.01125}.

\bibitem[Bleich et~al., 2014]{bleich2014variable}
Bleich, J., Kapelner, A., George, E.~I., and Jensen, S.~T. (2014).
\newblock Variable selection for {BART}: An application to gene regulation.
\newblock {\em The Annals of Applied Statistics}, 8(3):1750--1781.

\bibitem[Breiman, 2001]{breiman2001random}
Breiman, L. (2001).
\newblock Random forests.
\newblock {\em Machine Learning}, 45(1):5--32.

\bibitem[Chipman et~al., 2010]{chipman2010bart}
Chipman, H.~A., George, E.~I., and McCulloch, R.~E. (2010).
\newblock Bart: {B}ayesian additive regression trees.
\newblock {\em The Annals of Applied Statistics}, 4(1):266--298.

\bibitem[De~Jonge et~al., 2010]{de2010adaptive}
De~Jonge, R., Van~Zanten, J., et~al. (2010).
\newblock Adaptive nonparametric bayesian inference using location-scale
  mixture priors.
\newblock {\em The Annals of Statistics}, 38(6):3300--3320.

\bibitem[Dorie et~al., 2017]{dorie2017automated}
Dorie, V., Hill, J., Shalit, U., Scott, M., and Cervone, D. (2017).
\newblock Automated versus do-it-yourself methods for causal inference: Lessons
  learned from a data analysis competition.
\newblock {\em arXiv preprint arXiv:1707.02641}.

\bibitem[Freund et~al., 1999]{freund1999short}
Freund, Y., Schapire, R., and Abe, N. (1999).
\newblock A short introduction to boosting.
\newblock {\em Journal-Japanese Society For Artificial Intelligence},
  14(771-780):1612.

\bibitem[Friedman, 1991]{friedman1991multivariate}
Friedman, J.~H. (1991).
\newblock Multivariate adaptive regression splines.
\newblock {\em The Annals of Statistics}, pages 1--67.

\bibitem[Ghosal et~al., 2000]{ghosal2000convergence}
Ghosal, S., Ghosh, J.~K., and van~der Vaart, A.~W. (2000).
\newblock Convergence rates of posterior distributions.
\newblock {\em Annals of Statistics}, 28(2):500--531.

\bibitem[Ghosal and Van Der~Vaart, 2007]{ghosal2007posterior}
Ghosal, S. and Van Der~Vaart, A. (2007).
\newblock Posterior convergence rates of dirichlet mixtures at smooth
  densities.
\newblock {\em The Annals of Statistics}, 35(2):697--723.

\bibitem[Ghosal and van~der Vaart, 2017]{ghosal2017fundamentals}
Ghosal, S. and van~der Vaart, A. (2017).
\newblock {\em Fundamentals of Nonparametric {B}ayesian Inference}, volume~44.
\newblock Cambridge University Press.

\bibitem[Ghosal et~al., 2007]{ghosal2007convergence}
Ghosal, S., Van Der~Vaart, A., et~al. (2007).
\newblock Convergence rates of posterior distributions for noniid observations.
\newblock {\em The Annals of Statistics}, 35:192--223.

\bibitem[Ghosal and Van Der~Vaart, 2001]{ghosal2001entropies}
Ghosal, S. and Van Der~Vaart, A.~W. (2001).
\newblock Entropies and rates of convergence for maximum likelihood and bayes
  estimation for mixtures of normal densities.
\newblock {\em Annals of Statistics}, pages 1233--1263.

\bibitem[Green, 1995]{green1995reversible}
Green, P.~J. (1995).
\newblock Reversible jump {M}arkov chain {M}onte {C}arlo computation and
  {B}ayesian model determination.
\newblock {\em Biometrika}, 82(4):711--732.

\bibitem[Gr{\"u}nwald, 2012]{grunwald2012safe}
Gr{\"u}nwald, P. (2012).
\newblock The safe {B}ayesian.
\newblock In {\em International Conference on Algorithmic Learning Theory},
  pages 169--183. Springer.

\bibitem[Gy{\"o}rfi et~al., 2006]{gyorfi2006distribution}
Gy{\"o}rfi, L., Kohler, M., Krzyzak, A., and Walk, H. (2006).
\newblock {\em A distribution-free theory of nonparametric regression}.
\newblock Springer Science \& Business Media.

\bibitem[Hahn et~al., 2017]{hahn2017bayesian}
Hahn, P.~R., Murray, J.~S., and Carvalho, C.~M. (2017).
\newblock Bayesian regression tree models for causal inference: regularization,
  confounding, and heterogeneous effects.
\newblock {\em arXiv prepring arXiv:1706.09523}.

\bibitem[Hastie et~al., 2009]{hastie2016elements}
Hastie, T., Ribshirani, R., and Friedman, J. (2009).
\newblock {\em The Elements of Statistical Learning}.
\newblock Springer, 2nd edition.

\bibitem[Hill, 2011]{hill2011bayesian}
Hill, J.~L. (2011).
\newblock Bayesian nonparametric modeling for causal inference.
\newblock {\em Journal of Computational and Graphical Statistics}, 20(1).

\bibitem[Hill, 2016]{Hill:2016}
Hill, J.~L. (2016).
\newblock {Atlantic Causal Inference Conference Competition results}.
\newblock Accessed May 27, 2017 at \url{
  http://jenniferhill7.wixsite.com/acic-2016/competition }.

\bibitem[Irsoy et~al., 2012]{irsoy2012soft}
Irsoy, O., Y{\i}ld{\i}z, O.~T., and Alpayd{\i}n, E. (2012).
\newblock Soft decision trees.
\newblock In {\em Proceedings of the International Conference on Pattern
  Recognition}, pages 1819--1822.

\bibitem[Kapelner and Bleich, 2016]{kapelner2014bartmachine}
Kapelner, A. and Bleich, J. (2016).
\newblock {bartMachine}: Machine learning with {B}ayesian additive regression
  trees.
\newblock {\em Journal of Statistical Software}, 70(4):1--40.

\bibitem[Kim et~al., 2007]{kim2007visualizable}
Kim, H., Loh, W.-Y., Shih, Y.-S., and Chaudhuri, P. (2007).
\newblock Visualizable and interpretable regression models with good prediction
  power.
\newblock {\em IIE Transactions}, 39(6):565--579.

\bibitem[Linero, 2016]{linero2016bayesian}
Linero, A.~R. (2016).
\newblock Bayesian regression trees for high dimensional prediction and
  variable selection.
\newblock {\em Journal of the American Statistical Association}.
\newblock To appear.

\bibitem[Miller and Dunson, 2018]{miller2018robust}
Miller, J.~W. and Dunson, D.~B. (2018).
\newblock Robust {B}ayesian inference via coarsening.
\newblock {\em Journal of the American Statistical Association},
  (accepted):1--31.

\bibitem[Murray, 2017]{murray2017log}
Murray, J.~S. (2017).
\newblock Log-linear {B}ayesian additive regression trees for categorical and
  count responses.
\newblock {\em arXiv preprint arXiv:1701.01503}.

\bibitem[{Rockova} and {van der Pas}, 2017]{rockova2017posterior}
{Rockova}, V. and {van der Pas}, S. (2017).
\newblock Posterior concentration for {B}ayesian regression trees and their
  ensembles.
\newblock {\em {arXiv preprint arXiv:1078.08734}}.

\bibitem[Shen et~al., 2013]{shen2013adaptive}
Shen, W., Tokdar, S.~T., and Ghosal, S. (2013).
\newblock Adaptive {B}ayesian multivariate density estimation with {D}irichlet
  mixtures.
\newblock {\em Biometrika}, 100(3):623--640.

\bibitem[Sparapani et~al., 2016]{sparapani2016nonparametric}
Sparapani, R.~A., Logan, B.~R., McCulloch, R.~E., and Laud, P.~W. (2016).
\newblock Nonparametric survival analysis using bayesian additive regression
  trees ({BART}).
\newblock {\em Statistics in medicine}.

\bibitem[Vehtari et~al., 2015]{vehtari2015practical}
Vehtari, A., Gelman, A., and Gabry, J. (2015).
\newblock Practical bayesian model evaluation using leave-one-out
  cross-validation and waic.
\newblock {\em arXiv preprint arXiv:1507.04544}.

\bibitem[Walker and Hjort, 2001]{walker2001bayesian}
Walker, S. and Hjort, N.~L. (2001).
\newblock On {B}ayesian consistency.
\newblock {\em Journal of the Royal Statistical Society: Series B (Statistical
  Methodology)}, 63:811--821.

\bibitem[Watanabe, 2013]{watanabe2013widely}
Watanabe, S. (2013).
\newblock A widely applicable bayesian information criterion.
\newblock {\em Journal of Machine Learning Research}, 14:867--897.

\bibitem[Yang and Dunson, 2014]{yang2014minimax}
Yang, Y. and Dunson, D.~B. (2014).
\newblock Minimax optimal bayesian aggregation.
\newblock {\em arXiv preprint arXiv:1403.1345}.

\bibitem[Yang and Tokdar, 2015]{yang2015minimax}
Yang, Y. and Tokdar, S.~T. (2015).
\newblock Minimax-optimal nonparametric regression in high dimensions.
\newblock {\em The Annals of Statistics}, 43(2):652--674.

\bibitem[Yuan and Lin, 2006]{yuan2006model}
Yuan, M. and Lin, Y. (2006).
\newblock Model selection and estimation in regression with grouped variables.
\newblock {\em Journal of the Royal Statistical Society: Series B (Statistical
  Methodology)}, 68(1):49--67.

\end{thebibliography}



\section*{Supplementary material}

# Supplementary Material to Soft Bayesian Additive Regression Trees: Ensembles that Adapt to Smoothness and Sparsity

Antonio R. Linero* and Yun Yang[†]

June 29, 2018

## S.1   MCMC details

To construct Metropolis-Hastings proposals for $\mathcal{T}_t$ and $\tau_t$, we require the likelihood of $Y^{(t)} = (Y_1^{(t)}, \ldots, Y_n^{(t)})$ after marginalizing over $\mathcal{M}_t$, where

$$Y_i^{(t)} = Y_i - \sum_{k \neq t} g(X_i; \mathcal{T}_k, \mathcal{M}_k).$$

Letting $M_t$ denote this marginal likelihood, it can be shown that the full conditional distribution of $(\mathcal{T}_t, \tau_t)$ is

$$\pi(\mathcal{T}_t, \tau_t \mid -) \propto M_t \pi(\mathcal{T}_t) \pi(\tau_t).$$

Regarding $t$ as fixed, let $\phi_{i\ell} = \phi(X_i; \mathcal{T}_t, \ell)$ and let $\phi_i = (\phi_{i1}, \ldots, \phi_{iL_t})^\top$. The likelihood of $Y^{(t)}$ before marginalizing $\mathcal{M}_t$ is

$$\prod_{i=1}^{n} (2\pi\sigma^2)^{-1/2} \exp\left\{ -\frac{1}{2\sigma^2} \left[ \sum_{\ell=1}^{L_t} \phi_{i\ell}(\mu_\ell - y_i) \right]^2 \right\}$$

*Department of Statistics, Florida State University, Email: arlinero@stat.fsu.edu
[†]Department of Statistics, Florida State University, Email: yyang@stat.fsu.edu

Given $\mathcal{T}_t$, the $\mu_\ell$'s have Normal$(0, \sigma_\mu^2/T)$ distributions; using conjugacy and the fact that $\sum_\ell \phi_{i\ell} = 1$, standard calculations give the marginal likelihood

$$M_t = \frac{|2\pi\Omega|^{1/2}}{(2\pi\sigma^2)^{n/2}|2\pi\sigma_\mu^2\,\mathrm{I}\,|^{1/2}} \exp\left(-\frac{\|y\|^2}{2\sigma^2} + \frac{1}{2}\widehat{\mu}^\top \Omega^{-1} \widehat{\mu}\right),$$

where

$$\Omega = \left(\frac{\sigma_\mu^2}{T}\mathrm{I} + \Lambda\right)^{-1}, \qquad \Lambda = \sum_{i=1}^n \phi_i \phi_i^\top / \sigma^2, \qquad \widehat{\mu} = \Omega \sum_{i=1}^n Y_i \phi_i / \sigma^2.$$

Additionally, it can be shown that the full conditional of $\mathcal{M}_t$ is Normal$(\widehat{\mu}, \Omega)$.

We use a random-walk Metropolis-Hastings proposal for $\tau_t$, with proposal $\log \tau_t' = \log \tau_t + U$ where $U \sim \mathrm{Uniform}(-1, 1)$. The transition density is given by $q(\tau' \to \tau) = 0.5\tau^{-1}$, so that the acceptance probability is

$$A(\tau \to \tau') = \frac{M_t' \pi(\tau') \tau'}{M_t \pi(\tau) \tau} \wedge 1.$$

The Metropolis-Hastings steps for modifying the tree structure are given in detail by Kapelner and Bleich (2016). The Metropolis-Hastings acceptance probability is

$$A(\mathcal{T}_t \to \mathcal{T}_t') = \frac{M_t' \pi(\mathcal{T}_t') q(\mathcal{T}_t' \to \mathcal{T}_t)}{M_t \pi(\mathcal{T}_t) q(\mathcal{T}_t \to \mathcal{T}_t')} \wedge 1.$$

The `Birth` step involves the following steps.

1. Select a leaf node $\ell$ to become a branch.

2. Select a predictor $j$ to construct the split.

3. Sample $C_b \sim \mathrm{Uniform}(a, b)$, with $(a, b)$ chosen so that the split is not redundant.

By regarding each leaf node as having "inactive" values of $(j, C_b)$, steps 2 and 3 can be regarded as retrospective sampling steps so that the transition to $(j, C_b)$ can be regarded has having

probability 1. Hence, the transition probability associated with this move is

$$q(\mathcal{T}_t \to \mathcal{T}_t') = \frac{p_{\text{Birth}}(\mathcal{T}_t)}{L_t},$$

where $p_{\text{Birth}}(\mathcal{T}_t)$ is the probability that a Birth step was proposed; note that, if $\mathcal{T}_t$ consists only of a root node, then a Death step is impossible, so that $p_{\text{Birth}}(\mathcal{T}_t)$ is not constant as a function of $\mathcal{T}_t$.

To transition back from the proposed tree requires sampling the selected leaf node in a Death step, which occurs with probability

$$q(\mathcal{T}_t' \to \mathcal{T}_t) = \frac{p_{\text{Death}}(\mathcal{T}_t')}{B' + 1},$$

where $p_{\text{Death}}(\mathcal{T}_t')$ is the probability that a Death step is proposed, and $B'$ is the number of branches whose children are both leaves, i.e., $B'$ is the number of branches which are not grandparents of any nodes. Finally, tree structure ratio can be computed to be

$$\frac{\pi(\mathcal{T}_t')}{\pi(\mathcal{T}_t)} = \frac{q(d_\ell)(1 - q(d_{\ell+1}))^2}{1 - q(d_\ell)}$$

The Death step involves the following steps.

1. Select a branch node $b$, which is not a grandparent.

2. Delete the two child nodes.

By the retrospective sampling argument above, the forward transition probability

$$q(\mathcal{T}_t \to \mathcal{T}_t') = \frac{p_{\text{Death}}(\mathcal{T}_t)}{B'},$$

while the backwards transition probability is

$$q(\mathcal{T}_t' \to \mathcal{T}_t) = \frac{p_{\text{Birth}}(\mathcal{T}_t')}{L_t - 1}.$$

The tree structure ratio is given by

$$\frac{\pi(\mathcal{T}_t')}{\pi(\mathcal{T}_t)} = \frac{1 - q(d_b)}{q(d_b)(1 - q(d_{b+1}))^2}.$$

Finally, the `Change` step involves

1. Select a branch node which is not a grandparent.

2. Select a new predictor $j$.

3. Sample a new cut point $C_b$

When proposing a `Change` step, Kapelner and Bleich (2016) note that massive cancelation occurs, which leads to the transition probability

$$A(\mathcal{T}_t \to \mathcal{T}_t') = \frac{M_t'}{M_t} \wedge 1.$$

We now describe the updates used for $(\sigma, \sigma_\mu, a)$. Let $R = (R_1, \ldots, R_n)$ be the residuals, $R_i = Y_i - f(X_i)$. Note that, under a flat prior for $\sigma^{-2}$, the full-conditional of $\sigma^{-2}$ is $\text{Ga}(N/2 + 1, \|R\|^2/2)$. We use this full-conditional under the flat prior as a proposal distribution for $\sigma^{-2}$; after adjusting for the Jacobian of the transformation, the acceptance probability becomes

$$A(\sigma \to \sigma') = \frac{\text{Cauchy}_+(\sigma' \mid 0, \widehat{\sigma}_{\text{lasso}})\sigma'^3}{\text{Cauchy}_+(\sigma \mid 0, \widehat{\sigma}_{\text{lasso}})\sigma^3} \wedge 1.$$

A similar strategy is used to update $\sigma_\mu$. Let $\boldsymbol{\mu}$ be the collection of leaf parameters across all trees, let $L$ denote the number of total leaves, and define $s_\mu = \sigma_\mu/\sqrt{T}$. We propose $s_\mu'^{-2} \sim \text{Ga}(L/2 + 1, \|\boldsymbol{\mu}\|^2/2)$ and accept this proposal with probability

$$A(\sigma_\mu \to \sigma_\mu') = \frac{\text{Cauchy}_+(s_\mu' \mid 0, 0.25/\sqrt{T})s_\mu'^3}{\text{Cauchy}_+(s_\mu \mid 0, 0.25/\sqrt{T})s_\mu^3} \wedge 1.$$

Lastly, the parameter $a$ can be updated by sampling $\rho = a/(a + \lambda_a)$ by slice sampling (Neal,

2003); the full conditional is

$$\pi(\rho \mid -) \propto \frac{\Gamma(a_\rho)}{\Gamma(a_\rho/P)^P} \exp\left(\lambda a_\rho \cdot \frac{1}{P} \sum_{j=1}^{P} \log s_j\right) \operatorname{Be}(\rho \mid a_a, b_a)$$

where $a_\rho = \lambda_a \rho/(1-\rho)$.

## S.2 Results for the Usual Posterior

In this section, we prove that, under more technical conditions on the prior (Assumption SP below), the usual posterior (or fractional posterior with $\eta = 1$) has convergence rate at least a multiple of $\varepsilon_n$ defined in Theorem 3, that is,

$$\Pi_{n,1}\Big[\|f - f_0\|_n \geq M\,\varepsilon_n\Big] \to 0, \quad \text{in probability as } n, p \to \infty, \tag{S.1}$$

where $\varepsilon_n = n^{-\alpha/(2\alpha+d)}(\log n)^t + \sqrt{n^{-1} d \log p}$ for any $t \geq \alpha(d+1)/(2\alpha+d)$. The rest of this section will be devoted to a proof of this statement.

Recall that each soft sum of trees regression function takes the form of

$$f(x) = \sum_{t=1}^{T} g(x; \mathcal{T}_t, \mathcal{M}_t),$$

where $\mathcal{T}_t$ specifies the tree topology (splitting directions and locations, and bandwidth parameter $\tau_b$ associated with each branch $b$ in the splitting tree) of the $t$'th tree, and $\mathcal{M}_t = (\mu_{t1}, \ldots, \mu_{tL_t})$ collects all the parameters for the lead nodes in the $t$'th tree. In particular, the $t$'th component $g(\cdot; \mathcal{T}_t, \mathcal{M}_t)$ takes the form of $g(x; \mathcal{T}_t, \mathcal{M}_t) = \sum_{\ell=1}^{L} \mu_{t\ell}\,\phi(x; \mathcal{T}_t, \ell)$, with

$$\phi(x; \mathcal{T}_t, \ell) = \prod_{b \in A(\ell)} \psi(x; \mathcal{T}, b)^{1-R_b}(1 - \psi(x; \mathcal{T}, b))^{R_b},$$

where $A(\ell)$ is the set of ancestor nodes of leaf $\ell$ and $R_b = 1$ if the path to $\ell$ goes right at $b$.

Recall that $\Pi$ denotes the prior measure. From Ghosal et al. (2000, 2007), it suffices to verify

the following three conditions to find a posterior convergence rate: for some sieve $\{\mathcal{F}_n\}_{n=1}^{\infty}$ and constant $K > 0$,

$$\Pi(\|f - f_0\|_\infty \leq \varepsilon_n) \geq e^{-n\varepsilon_n^2},$$

$$\Pi(f \notin \mathcal{F}_n) \leq e^{-4n\varepsilon_n^2}, \tag{S.2}$$

$$\log N(\bar{\varepsilon}_n, \mathcal{F}_n, \|\cdot\|_\infty) \leq n\bar{\varepsilon}_n^2,$$

where recall that $\|\cdot\|_\infty$ is the sup-norm and $f_0$ is the truth. Under these three conditions, the posterior contraction rate would be at least $\max\{\varepsilon_n, \bar{\varepsilon}_n\}$ under the $\mathcal{L}^2(\mathbb{P}_n)$ norm $\|\cdot\|_n$. We make the following assumptions about the prior $\Pi$ to simply the proof.

**Assumption SP (Stronger prior conditions):**

1. The prior on the number of trees satisfies $\Pi(T \geq k) \leq e^{-ck}$ for all $k \geq 0$. In addition, (P1) in Assumption P holds.

2. There is a single bandwidth parameter $\tau_t \equiv \tau$ for each tree $t = 1, \ldots, T$, and its prior satisfies $\Pi(\tau \geq x) \leq a_1 e^{-x^{a_2}}$ and $\Pi(\tau^{-1} \geq x) \leq b_1 e^{-x^{b_2}}$ for all sufficiently large $x > 0$ and some constants $a_1 > 0, b_1 > 0, 0 < a_2 < 1, 0 < b_2 < 1$. Moreover, the density function of $\tau^{-1}$ satisfies $\pi_{\tau^{-1}}(x) \geq a_3 e^{-b_3 x}$ for all large enough $x$ and some positive constants $a_3, b_3$.

3. Condition (P3) in Assumption P holds.

4. The prior distribution on the i.i.d leaf values $\mu$ satisfies $\Pi(|\mu| \geq t) \leq c_1 e^{-t^{c_2}}$ for all $t \geq 0$ and some positive constants $c_1, c_2$. In addition, (P4) in Assumption P holds.

5. The prior on the depth $D_t$ of a tree is truncated so that $\Pi(D_t > d_0) = 0$ for some $d_0 \geq d$. In addition, (P5) in Assumption P holds.

Condition 1 holds for a geometric prior on $T$, and condition 2 holds when $\tau$ has an inverse-Gamma distribution truncated to some neighborhood of 0. Condition 4 holds when the $\mu_{t\ell}$'s have Laplace tails, but can be weakened to allow for $\mu_{t\ell} \sim \text{Normal}(0, \sigma_\mu^2)$ with an exponential prior on $\sigma_\mu^2$. The proof of Theorem 3 verifies the first condition with $\varepsilon_n$ being a multiple of $n^{-\alpha/(2\alpha+d)}(\log n)^t + \sqrt{n^{-1}d\log p}$ for any $t \geq \alpha(d+1)/(2\alpha+d)$; the only change required is to

use the density lower bound for the single global $\tau$ rather than the tree-specific $\tau_t$ lower bound. It remains to verify the other two conditions.

Fix $\varepsilon > 0$, $\sigma_0 > 0$ and integers $H, M, d > 0$, and define

$$\mathcal{F} = \Bigg\{ f(\cdot) = \sum_{t=1}^{T} g(\cdot; \mathcal{T}_t, \mathcal{M}_t) : T \leq Cn\varepsilon^2, \text{ each tree } \mathcal{T}_t \text{ has depth at most } H,$$

the common bandwidth parameter $\tau$ satisfies $\sigma_1 \leq \tau^{-1} \leq \sigma_2$,

the total number of distinct splitting directions in $f$ is at most $d_{max}$,

for each $t \leq T$, $L_t \leq 2^H$, and $j \in \{1, \ldots, L_t\}$, $\mu_{tj} \in [-M, M] \Bigg\}.$

Let $N(\mathcal{F}, \delta)$ denote the $\delta$-covering number of $\mathcal{F}$ relative to the supreme norm, that is, the minimal size if $\delta$-net of the set $\mathcal{F}$ under the supreme norm metric.

**Lemma 1.** *For the above set $\mathcal{F}$, under Assumptions G and SP, we have*

1. *Covering entropy control:* $\log N(\mathcal{F}, D\varepsilon) \leq d_{max} \log p + 3Cn\varepsilon^2 \, 2^H \log\left(d_{max} \sigma_1^{-1} \sigma_2^2 Cn\varepsilon \, 2^H M\right)$ *for some constant $D$ depending only on $\psi$;*

2. *Complement probability bound: choose $H \geq d_0$, then* $\Pi(\mathcal{F}^c) \leq \exp\{-c\,Cn\varepsilon^2\} + 2^H Cn\varepsilon^2 \cdot \left[\exp\{-E\,d_{max}\log p\} + c_1 \exp\{-M^{c_2}\}\right] + a_1 \exp\{-\sigma_1^{-a_2}\} + b_1 \exp\{-\sigma_2^{b_2}\}$ *for some constant $E > 0$ depending only on hyperparameter $\xi > 1$ in the Dirichlet prior.*

Before proving this lemma, let us illustrate how the lemma leads to the claimed posterior convergence rate. In fact, by taking $\varepsilon_n = C_0\big(n^{-\alpha/(2\alpha+d)}(\log n)^t + \sqrt{n^{-1}d\log p}\big)$, $d_{max} = \min\{1, \lfloor C_1 n\varepsilon_n^2/\log p\rfloor\}$, $\varepsilon = C_2\varepsilon_n^2$, $\sigma_1^{-a_2} = C_3 n\varepsilon_n^2$, $\sigma_2^{b_2} = C_4 n\varepsilon_n^2$, and $M^{c_2} = C_5 n\varepsilon_n^2$, where constants $C_j(k = 0, 1, \ldots, 5)$ are taken sufficiently large, we can deduce from Lemma 1 that there exists a set $\mathcal{F}$ such that

$$\Pi(f \notin \mathcal{F}) \leq e^{-4n\varepsilon_n^2}, \quad \text{and} \quad \log N(\bar{\varepsilon}_n, \mathcal{F}, \|\cdot\|_\infty) \leq n\bar{\varepsilon}_n^2,$$

where $\bar{\varepsilon}_n$ is a multiple of $\varepsilon_n$. Combining with this result with the proof of Theorem 3 leads to

a proof of the three conditions in (S.2), thereby a proof of the desired posterior convergence in (S.1).

In the remainder of this section, we prove Lemma 1.

*Proof of Lemma 1. Bound on covering entropy.* First, let us consider the covering entropy of the following space of a single component with splitting only along the first $d_{max}$-coordinates ($p$ coordinates in total),

$$\mathcal{F}_S = \Big\{ g(\cdot; \mathcal{T}, \mathcal{M}) : \text{tree } \mathcal{T} \text{ has depth at most } H, \text{ all splits are along the first } d\text{-coordinates},$$

$$\text{the common bandwidth parameter } \tau \text{ satisfies } \sigma_1 \leq \tau^{-1} \leq \sigma_2,$$

$$\text{for each } j \in \{1, \ldots, L\}, L \leq 2^H, \text{ and } \mu_j \in [-M, M], \text{ where } \mathcal{M} = (\mu_1, \ldots, \mu_L) \Big\}.$$

We first claim that an upper bound of $\log N(\mathcal{F}_S, D(Cn\varepsilon)^{-1})$ implies an upper bound of $\log N(\mathcal{F}, D\varepsilon)$ through the relation

$$\log N(\mathcal{F}, D\varepsilon) \leq d_{max} \log p + Cn\varepsilon^2 \log N(\mathcal{F}_S, D(Cn\varepsilon)^{-1}).$$

To see this, note first that there are at most $p^{d_{max}}$ many different subsets of $\{1, \ldots, p\}$ with size at most $d_{max}$; second, under the constraint that all additive components $g(\cdot; \mathcal{T}_t, \mathcal{M}_t)$'s are splitting along the same $d_{max}$ coordinates (for example, along the first $d$-coordinates), an $\varepsilon$-net of all such functions in $\mathcal{F}$ can be constructed as the direct sum of $Cn\varepsilon^2$-copies of a $D(Cn\varepsilon)^{-1}$-net of $\mathcal{F}_S$ (the approximation error is at most $T \cdot D(Cn\varepsilon)^{-1} \leq D\varepsilon$), which has cardinality at most $\big[ N(\mathcal{F}_S, D(Cn\varepsilon)^{-1}) \big]^{Cn\varepsilon^2}$; finally, a $D\varepsilon$-net of $\mathcal{F}$ can be formed as a union of all such $D\varepsilon$-nets in the second step over all different subsets of $\{1, \ldots, p\}$ as splitting directions. This $D\varepsilon$-net of $\mathcal{F}$ has cardinality at most $p^{d_{max}} \times \big[ N(\mathcal{F}_S, D(Cn\varepsilon)^{-1}) \big]^{Cn\varepsilon^2}$, which leads the preceding displayed inequality.

Consequently, it remains to show $\log N(\mathcal{F}_S, D(Cn\varepsilon)^{-1}) \leq 3 \cdot 2^H \log \big( d_{max} \sigma_1^{-1} \sigma_2^2 Cn\varepsilon \, 2^H M \big)$. Let $\hat{R}$ be a $\sigma_2^{-1}(Cn\varepsilon)^{-1} H^{-1} 2^{-H}$-net of $[0,1]$ (for the splitting locations), $\hat{\Gamma}$ be a $\sigma_2^{-2}(Cn\varepsilon)^{-1} H^{-1} 2^{-H}$-

net of $[0, \sigma_1^{-1}]$ (for bandwidth $\tau$), and $\hat{M}$ be a $(Cn\varepsilon)^{-1}2^{-H}$-net of $[-M, M]$ (for the leaf values). Then $\hat{R}$ can be chosen to have size at most $\sigma_2 Cn\varepsilon H 2^H$, $\hat{\Gamma}$ size at most $\sigma_1^{-1}\sigma_2^2 Cn\varepsilon H 2^H$, $\hat{\Gamma}$, and $\hat{M}$ size at most $2MCn\varepsilon 2^H$.

Given any $g(\cdot; \mathcal{T}, \mathcal{M}) = \sum_{\ell=1}^{L} \mu_l \phi(x; \mathcal{T}, \ell) \in \mathcal{F}_S$, we can always expand the tree such that it has depth $H$ with $L = 2^H$ leaves (adding zero values to the expanded leaves). Let $x = (x_1, \ldots, x_{2^{H+1}-1})$ denote a list of all $1+2+\cdots+2^H$ splitting locations, $\tau$ the common bandwidth parameter, and $\mu = (\mu_1, \ldots, \mu_L)$ a list of all leaf values. According to the constructions of the nets, we can always find a $\hat{g}(\cdot; \hat{\mathcal{T}}, \hat{\mathcal{M}}) = \sum_{\ell=1}^{L} \hat{\mu}_l \phi(x; \hat{\mathcal{T}}, \ell)$ whose splitting directions at all branches exactly match those of $g(\cdot; \mathcal{T}, \mathcal{M})$, and whose splitting locations $\hat{x}$, bandwidth parameters $\hat{\tau}$ and leaf values $\hat{\mu}$ satisfy

$$|x_j - \hat{x}_j| \leq \sigma_2^{-1}(Cn\varepsilon)^{-1}H^{-1}2^{-H}, \quad \text{for all } j = 1, \ldots, 2^{H+1}-1,$$

$$|\tau - \hat{\tau}| \leq \sigma_2^{-2}(Cn\varepsilon)^{-1}H^{-1}2^{-H},$$

$$|\mu_\ell - \hat{\mu}_\ell| \leq (Cn\varepsilon)^{-1}2^{-H}, \quad \text{for all } \ell = 1, \ldots, L.$$

First, we note the following perturbation bound for the gating function $\psi$ (under Assumption G)

$$\left\|\psi\left(\frac{\cdot - x}{\tau}\right) - \psi\left(\frac{\cdot - x'}{\tau'}\right)\right\|_\infty \leq C_\psi \left(\frac{|\tau - \tau'|}{\min\{\tau^2, (\tau')^2\}} + \frac{|x - x'|}{\min\{\tau, \tau'\}}\right), \quad \text{for all } x, x' \in \mathbb{R}, \tau, \tau' > 0,$$

where $C_\psi$ is some constant depending on $\|\psi'\|_\infty$. This perturbation bound implies, for any $j = 1, \ldots, 2^{H+1} - 1$ and some constant $C_1 > 0$,

$$\left\|\psi\left(\frac{\cdot - x_j}{\tau}\right) - \psi\left(\frac{\cdot - \hat{x}_j}{\hat{\tau}}\right)\right\|_\infty \leq C_1(Cn\varepsilon)^{-1}H^{-1}2^{-H}.$$

A combination of this approximation error bound with the fact that $\psi \in [0, 1]$ implies for each

leaf $\ell$

$$\left\|\phi(x; \mathcal{T}, \ell) - \phi(x; \hat{\mathcal{T}}, \ell)\right\|_\infty$$
$$= \left\|\prod_{b \in A(\ell)} \psi(\cdot; \mathcal{T}, b)^{1-R_b}(1 - \psi(\cdot; \mathcal{T}, b))^{R_b} - \prod_{b \in A(\ell)} \psi(\cdot; \hat{\mathcal{T}}, b)^{1-R_b}(1 - \psi(\cdot; \hat{\mathcal{T}}, b))^{R_b}\right\|_\infty$$
$$\leq \sum_{h=1}^{H} C_1 (Cn\varepsilon)^{-1} H^{-1} = C_1 (Cn\varepsilon)^{-1} 2^{-H}.$$

where we have used the fact that $\mathcal{T}$ and $\hat{\mathcal{T}}$ share the same splitting directions, and have applied the inequality

$$\left|\prod_{j=1}^{N} a_j - \prod_{j=1}^{N} b_j\right| \leq \sum_{j=1}^{N} |a_j - b_j|$$

for all numbers $a_j \in [0,1]$, $b_j \in [0,1]$, $j = 1, \ldots, N$, and $N \geq 1$. Finally, by adding and subtracting the same term we obtain ($L = 2^H$)

$$\left\|g(\cdot; \mathcal{T}, \mathcal{M}) - \hat{g}(\cdot; \hat{\mathcal{T}}, \hat{\mathcal{M}})\right\|_\infty$$
$$= \left\|\sum_{\ell=1}^{L} \mu_\ell \, \phi(x; \mathcal{T}, \ell) - \sum_{\ell=1}^{L} \hat{\mu}_\ell \, \phi(x; \hat{\mathcal{T}}, \ell)\right\|_\infty$$
$$\leq \left\|\sum_{\ell=1}^{L} (\mu_\ell - \hat{\mu}_\ell)\, \phi(x; \mathcal{T}, \ell)\right\|_\infty + \left\|\sum_{\ell=1}^{L} \hat{\mu}_l \left(\phi(x; \mathcal{T}, \ell) - \phi(x; \hat{\mathcal{T}}, \ell)\right)\right\|_\infty$$
$$\leq 2L(Cn\varepsilon)^{-1} 2^{-H} + 2LC_1(Cn\varepsilon)^{-1} 2^{-H} = 2(C_1 + 1)(Cn\varepsilon)^{-1}.$$

Therefore, a $2(C_1 + 1)(Cn\varepsilon)^{-1}$-net of $\mathcal{F}_S$ under the supreme norm can be constructed with all $\hat{g}(\cdot; \hat{\mathcal{T}}, \hat{\mathcal{M}})$ above. In addition, the total number of such $\hat{g}(\cdot; \hat{\mathcal{T}}, \hat{\mathcal{M}})$ can be bounded above by a multiple of $d_{max}^{2^H} (\sigma_2 \, Cn\varepsilon H 2^H)^{2^{H+1}} (\sigma_1^{-1} \sigma_2^2 \, Cn\varepsilon H 2^H)^{2^{H+1}} (2MCn\varepsilon 2^H)^{2^H}$ (the first factor is because each branch can choose one of the $d$ directions to split along). This proves the upper bound on the covering entropy of $\mathcal{F}_S$, which leads to the claimed bound in part 1.

*Complement probability bound.* We apply a union bound argument to obtain (and use the prior

independence)

$$\Pi(\mathcal{F}^c) \le \Pi(T > Cn\varepsilon^2) + \Pi(T \le Cn\varepsilon^2) \cdot Cn\varepsilon^2 \cdot 2^H \cdot \Pi(\mu_{tj} \notin [-M, M]) + \Pi(\tau^{-1} < \sigma_1 \text{ or } \tau^{-1} > \sigma_2)$$
$$+ \Pi(T \le Cn\varepsilon^2, \text{total number of distinct splitting directions is at most } d_{max}).$$

The first two terms can be bounded using the prior tail conditions, and it remains to bound the last term. To that end, recall that $s = (s_1, \ldots, s_p) \sim \mathcal{D}(a/p^\xi, \ldots, a/p^\xi)$ is the splitting proportion vector. Let $s_{(1)} \ge s_{(2)} \ge \cdots \ge s_{(p)}$ denote the reordering of components of $s$. The proof of Lemma 5.4 in Yang and Dunson (2014) implies that for any $\kappa > 0$ and some constant $c_d$,

$$\Pi\Big(\sum_{j=d_{max}+1}^{p} s_{(j)} \ge \kappa\Big) \le \exp\big\{-c_d(\xi-1)\,d_{max}\log p + \log(1/\kappa)\big\}.$$

Consequently, a union bound implies that the probability that all splits (at most $Cn\varepsilon^2 2^H$ splits under the event $\{T \le Cn\varepsilon^2\}$) are along the indices corresponding to the largest $d_{max}$ components of $s$ is at most (choose $\kappa = 1/(Cn\varepsilon^2 2^H) \exp\{-c_d(\xi-1)d_{max}\log p/2\}$)

$$Cn\varepsilon^2 2^H \kappa + \exp\big\{-c_d(\xi-1)\,d_{max}\log p + \log(1/\kappa)\big\} \le 2\exp\big\{-\frac{1}{2}c_d(\xi-1)\,d_{max}\log p + \log(Cn\varepsilon^2 2^H)\big\}.$$

Putting all pieces together yields the claimed complement probability upper bound.

$\square$

## S.3  Timing Comparisons and Potential Improvements

In this section, we examine the runtime of the SBART algorithm relative to the competitors of Section 4. We emphasize that the implementation of SBART used in this paper has not been optimized, and that substantial improvements to our implementation are possible. The software used in this paper is undergoing active development, and we hope in the future to

incorporate the modifications mentioned below.

The primary bottleneck in terms of speed, relative to BART, is the need to compute the sufficient statistics $(\Omega, \widehat{\mu})$ described in Section S.1. Our implementation computes these statistics $3T$ times for each iteration of our Gibbs sampler (twice when updating the tree topology and leaf parameters, once when updating $\tau$); an adequately designed joint update for the tree topology and $\tau$ could reduce this to $2T$ evaluations. Additionally, caching $\phi_i$'s from the previous iteration of the algorithm removes an additional likelihood evaluation, giving a total 66% decrease in the likelihood evaluation time, relative to our current implementation.

Additionally, BART benefits from the fact that it does not need to consider all the leaf nodes of the tree when computing the likelihood contribution of each observation; the analogous algorithm for BART features a diagonal $\Omega$ which can be computed very efficiently by caching the observations associated to each leaf node. To some extent, this could be addressed with the SBART algorithm by using a monotone function $\psi(x)$ such that $\psi(L) = 0$ and $\psi(U) = 1$ for some constants $(L, U)$. This would allow us to disregard entire branches of a tree whenever $\psi\{(x - C_b)/\tau_b\} = 0$ or $1$ and open up similar caching strategies that are utilized by efficient implementations of BART; a theoretical tradeoff for using such a $\psi$ is that the effective kernel $K(x) = \psi(x)\{1 - \psi(x)\}$ is no longer analytic. Lastly, we note that the likelihood computations required by SBART can be easily parallelized.

Table S.2 compares the time, in seconds, required by each algorithm to fit the `bbb` dataset, using "default" settings and cross-validation. Several remarks are warranted. First, `SoftBart` is the slowest package, being roughly half the speed `BayesTree`; as noted above, much of this gap can be closed by relative minor design changes. Consequently, we find these results encouraging. It is apparent that, if the ultimate goal is an implementation which is as fast as `dbarts` — which is almost as efficient as `randomForest` and `xgboost` — then much further work is required in optimizing `SoftBart`. Lastly, we note that `SoftBart` is more competitive with other packages when cross validation is used. This is because `SoftBart` only requires tuning $T$, while the other packages feature other tuning parameters. For example, `dbarts` considers hyperparameters

| Package | Default | CV |
|---|---|---|
| randomForest | 1.524 | 128.537 |
| glmnet | - | 0.273 |
| xgBoost | 1.029★ | 53.81★ |
| SoftBart | 27.324 | 705.598 |
| dbarts | 1.812 | 222.724★ |
| BayesTree | 12.854 | - |

Table S.2: Time, in seconds, to fit the `bbb` dataset using each method. Results with a ★ indicate that the package made use of parallel processing.

on a $3 \times 3 \times 3$ grid. In principle, `dbarts` could reduce this computational burden by placing priors on certain hyperparameters, as `SoftBart` does.

# References

Ghosal, S., Ghosh, J. K., and van der Vaart, A. W. (2000). Convergence rates of posterior distributions. *Annals of Statistics*, 28(2):500–531.

Ghosal, S., Van Der Vaart, A., et al. (2007). Convergence rates of posterior distributions for noniid observations. *The Annals of Statistics*, 35:192–223.

Kapelner, A. and Bleich, J. (2016). bartMachine: Machine learning with Bayesian additive regression trees. *Journal of Statistical Software*, 70(4):1–40.

Neal, R. M. (2003). Slice sampling. *The Annals of Statistics*, 31:705–767.

Yang, Y. and Dunson, D. B. (2014). Minimax optimal bayesian aggregation. *arXiv preprint arXiv:1403.1345*.






\end{document}